\documentclass[12pt]{article}
\usepackage{amssymb}
\usepackage{amsmath}
\usepackage{authblk}

\usepackage{enumitem}
\usepackage{float}
\usepackage[pdftex,pdfstartview=FitH,pdfborderstyle={/S/B/W 1},%
    colorlinks=true, linkcolor=blue, urlcolor=blue, citecolor=blue,%
    pagebackref=true]{hyperref}
\usepackage{tikz}
\usetikzlibrary{arrows.meta}
\usetikzlibrary{trees}
\usepackage{pgfplots}
\usepackage[dvipsnames]{xcolor}

\usepackage{booktabs}
\usepackage{pifont} 
\pgfplotsset{compat=1.18}

\newtheorem{theorem}{Theorem}[section]
\newtheorem{definition}[theorem]{Definition}
\newtheorem{lemma}[theorem]{Lemma}
\newtheorem{corollary}[theorem]{Corollary}
\newtheorem{property}[theorem]{Property}
\newtheorem{procedure}[theorem]{Procedure}
\newenvironment{proof}[1][Proof]{\textbf{#1.} }{\ $\square$}

\newcommand{\abs}[1]{{\left|{#1}\right|}}
\newcommand{\pr}[1]{{\left({#1}\right)}}

\newcommand{\cbr}[1]{{\left\{{#1}\right\}}}
\newcommand{\clopen}[1]{\left[{#1}\right)}
\newcommand{\ceil}[1]{\left\lceil{#1}\right\rceil}
\newcommand{\floor}[1]{\left\lfloor{#1}\right\rfloor}
\newcommand{\Bin}[2]{\operatorname{Bin}\pr{{#1}, {#2}}}
\newcommand{\Exp}[1]{\operatorname{Exp}\pr{#1}}
\newcommand{\Pois}[1]{\operatorname{Pois}\pr{#1}}
\newcommand{\BetaVar}[2]{\operatorname{\beta}\pr{{#1}, {#2}}}

\newcommand{\bigo}[1]{O\pr{{#1}}}
\newcommand{\num}[1]{N({#1})}

\newcommand{\timeseg}[2]{\clopen{t_{#1}, t_{#2}}}
\newcommand{\rdatalist}[1]{F^{#1}}
\newcommand{\elh}[1]{\pr{e#1, l#1, u#1}}
\newcommand{\x}{\times}

\newcommand{\EE}{\mathbb{E}}
\newcommand{\NN}{\mathbb{N}}
\newcommand{\PP}{\mathbb{P}}
\newcommand{\RR}{\mathbb{R}}
\newcommand{\ZZ}{\mathbb{Z}}

\newcommand{\cG}{\mathcal{G}}

\newcommand{\cR}{\mathcal{R}}

\newcommand{\cX}{\mathcal{X}}
\newcommand{\cY}{\mathcal{Y}}

\newcommand{\fnm}[1]{#1}
\newcommand{\sur}[1]{#1}
\newcommand{\orgname}[1]{#1}
\newcommand{\orgdiv}[1]{#1}
\newcommand{\orgaddress}[1]{#1}
\newcommand{\postcode}[1]{#1}
\newcommand{\city}[1]{#1}
\newcommand{\state}[1]{#1}
\newcommand{\country}[1]{#1}

\title{An exact tau-leaping method}
\author[1]{\fnm{Ron} \sur{Solan}} 
\author[1, 2, 3]{\fnm{Gad} \sur{Getz}} 

\affil[1]{\orgname{Broad Institute of Massachusetts Institute of Technology and Harvard}, \orgaddress{\city{Cambridge}, \postcode{02142}, \state{MA}, \country{USA}}}
\affil[2]{\orgdiv{Krantz Center for Cancer Research and Department of Pathology}, \orgname{Mass General Hospital}, \orgaddress{\city{Boston}, \postcode{02115}, \state{MA}, \country{USA}}}
\affil[3]{\orgname{Harvard Medical School}, \orgaddress{\city{Boston}, \postcode{02115}, \state{MA}, \country{USA}}}

\begin{document}
\maketitle

\begin{abstract}
    The Gillespie algorithm and its extensions are commonly used for the 
    simulation of chemical reaction networks.
    A limitation of these algorithms is that they have to process and update the system after every reaction,
    requiring significant computation.
    Another class of algorithms, based on the $\tau$-leaping method, 
    is able to simulate multiple reactions at a time at the cost of decreased accuracy.
    We present a new algorithm for the exact simulation of chemical reaction networks
    that is capable of sampling multiple reactions at a time via a first-order approximation similarly
    to the $\tau$-leaping methods.
    We prove that the algorithm has an improved runtime complexity compared to existing methods
    for the exact simulation of chemical reaction networks, 
    and present an efficient and easy to use implementation that outperforms 
    existing methods in practice.
\end{abstract}

\section{Introduction}

Stochastic models are used in many branches of science 
to describe systems that change over time, 
where the change depends both on the current state of the system and on random external forces. 
The state of the system can be continuous, as when describing the motion of particles, 
or discrete, as when describing chemical reactions. 
Discrete state stochastic models have been used in, e.g., chemistry~\cite{gegenhuber2017fusing}, 
biology~\cite{weinberger2005stochastic}, 
ecology~\cite{dobrinevski2012extinction}, 
and astronomy~\cite{cuppen2013kinetic}.

The Gillespie algorithm~\cite{gillespie1977exact} is the best-known algorithm for the simulation of discrete state stochastic models.
The algorithm works by repeatedly sampling the next change in the system and the time until it occurs.
After sampling the change in the system, the algorithm updates the list of possible future changes and their rates.
It repeats these two steps until simulating the whole time period.

Variations on the Gillespie algorithm allow time-varying systems~\cite{anderson2007time_dependent}, 
delayed reactions~\cite{cai2007delayed}, 
and differentiation of functions on the state with respect to the network parameters~\cite{rijal2025differentiable}.


The current directions in optimizing the Gillespie algorithm focus on using efficient data structures to accelerate
the sampling~\cite{thanh2016rssacr} and updating~\cite{ghosh2021blsssa} steps.
This approach is limited: It has to perform work for every single state change, 
and it only looks forward a single event at a time.
This is in contrast to continuous SDE simulation methods that can employ higher derivatives, 
as in Runge-Kutta methods~\cite{burrage1996stochatic_runge_kutta}.
This limitation is especially detrimental when the system has different parts with different behaviors:
parts of the system with quick stochastic fluctuations can often be accurately and efficiently simulated
by continuous methods, but their frequent state changes increase the runtime of the Gillespie algorithm.
Other parts of the system, which change more slowly, 
require less simulation work from the Gillespie algorithm, 
but reduce the accuracy of continuous methods.
Previous solutions to this problem have focused on specific types of networks~\cite{cao2005slow} 
or use approximations~\cite{ahmadian2017hybrid}.

In this paper we present a new algorithm for the exact simulation of discrete state SDEs, 
which we call the \emph{$\tau$-splitting method}. 
This method overcomes the limitations of the Gillespie algorithm and bridges 
the gap between continuous and discrete simulation methods
by simulating the progress of the system using a first order approximation, 
and performing work not for every reaction, 
but for every time a reaction deviates from the continuous approximation.
This approach significantly reduces the runtime compared to other approaches,
and opens new avenues of improvement for stochastic simulation algorithms.

\section{The problem and the main result}

\subsection{Informal background}
\label{section:informal}

We are interested in the exact simulation of biochemical systems,
and in particular in simulating the dynamics in a biological cell.
In a cell, there are many different chemical reactions occurring in parallel,
and the amounts of molecules span several orders of magnitude: 
e.g., in a human neuron there are one or two DNA molecules encoding each gene and $500\cdot10^9$ potassium atoms~\cite{howard1993neuron_volume,purves2018neuroscience}.

We describe chemical systems following the chemical master equation (CME) 
formulation~\cite{gillespie1992rigorous}.
In these systems, there are molecules undergoing chemical reactions.
For a chemical species $A$, we denote the number of $A$ molecules by $\num A$.
We will denote a chemical reaction in which the molecules $A$ and $B$ combine to form the molecule $C$ by
\begin{align*}
    A + B \to C.
\end{align*}
When one of the sides is empty,\footnote{
    This can happen, e.g., when a molecule enters the cell from an environment 
    that is assumed to be constant, and is not modelled.
} we denote it by $\phi$ for visual clarity.
Each reaction has a rate constant $r$. 
According to the CME, this means that in every time period of length $dt$, 
the probability of the reaction occurring is $ r \num A \num B dt + \bigo{{dt}^2}$.
When the reaction $A + B \to C$ occurs, $\num A$ decreases by 1, $\num B$ decreases by 1, and $\num C$ increases by 1.
When there are reactions with two molecules of the same type:
\begin{align*}
    2A \to C,
\end{align*}
the probability of a reaction in a time period of length $dt$ is $r \binom{\num A}{2} dt + \bigo{{dt}^2}$, 
since the rate is proportional to the number of pairs of $A$ molecules
that can react with each other.
This formula is different from $r N^2\pr{A} dt$ when the number of $A$ molecules is small.
Indeed, when there is only one $A$ molecule, the reaction $2A \to C$ can't progress.

The CME defines a continuous-time Markov chain describing the progress of the chemical system.
The most basic algorithm for sampling the state of the system at some time $T$ 
is the Gillespie algorithm~\cite{gillespie1977exact}.
The Gillespie algorithm repeatedly samples the identity of the next reaction and the time until it occurs.
This limits its ability to simulate systems where a very large number of reactions take place.
To improve the performance of sampling the state of the system, 
the $\tau$-leaping method was introduced~\cite{gillespie2001approximate},
which uses techniques similar to ODE simulations.
However, the $\tau$-leaping method does not sample exactly from the 
distribution defined by the CME, 
but from an approximation of it.
In this paper we introduce the $\tau$-splitting algorithm, 
which builds on the ideas of the $\tau$-leaping method, but samples the state of the system 
from the exact distribution defined by the CME.
Like the $\tau$-leaping method, the $\tau$-splitting algorithm does not sample the whole trajectory of the system,
and instead skips over time periods using a constant-rate approximation.
Unlike the $\tau$-leaping method, the $\tau$-splitting algorithm only skips over periods 
where the constant-rate approximation is guaranteed to be accurate, 
and the distribution of the states sampled using the $\tau$-splitting algorithm is the 
same as the distribution sampled using the Gillespie algorithm.

We will now illustrate the Gillespie, the $\tau$-leaping, 
and the $\tau$-splitting algorithms for the simulation of the CME\@.
Consider the system of radioactive decay over the time period $[0, T]$.
This system has a single reaction,
\begin{align*}
    A \to B,
\end{align*}
where one $A$ atom decays to one $B$ atom.
Let $\gamma$ be the rate of the reaction.

\paragraph{The Gillespie algorithm}
In the Gillespie algorithm, we simulate the system iteratively.
When there are $\num A$ atoms of $A$, the rate is $\gamma \num A$. 
Thus, the time until the next decay is distributed as $\Exp{\gamma \num A}$.
In each iteration, 
we sample the time until the next radioactive decay as $\Exp{\gamma \num A}$, 
then subtract 1 from $\num A$ and update the rate to reflect the new $\num A$.
This advances the simulation one decay at a time.

\paragraph{The $\tau$-leaping method}
In the $\tau$-leaping method, 
we split the simulated time period to segments of length $\tau$.
For a given time segment of length $\tau$, 
let $\num A$ be the number of $A$ atoms at the beginning of the time segment.
We assume that the decay rate during the time segment is constant and equal to the rate at the beginning, 
and sample the number of radioactive decays via $\Pois{\gamma \num A\tau}$, 
the Poisson random variable with rate $\gamma \num A \tau$.

Figure~\ref{figure:tau_leap} displays the decay of $A$ atoms simulated 
by the Gillespie algorithm (in black) and the $\tau$-leaping method (in cyan).
Since the Gillespie algorithm samples one reaction at a time, 
it simulates every single change in $\num A$. 
Since the $\tau$-leaping method splits the time period into segments
of length $\tau$ and simulates every time segment without simulating
the exact trajectory within each segment, 
the cyan arrows skip over the individual reactions and point directly at
the state in the end of each segment.
Since the $\tau$-leaping method only updates the reaction rate
at the end of each segment, it does not take into account that 
the rate decreases during each segment, 
and overestimates the number of decays.

\begin{figure}[H]
    \centering
    
    \begin{tikzpicture}
        \draw[->] (0,0) -- (7,0); 
        \draw[->] (0,0) -- (0,6); 
        
        \node[left] at (0, 3) {Amount};
        \node[below] at (3.5, -0.5) {Time};

        \node[below] at (0, 0) {0};
        \node[below] at (1, 0) {$\tau$};
        \node[below] at (2, 0) {$2\tau$};
        \node[below] at (3, 0) {$3\tau$};
        \node[below] at (4, 0) {$4\tau$};
        \node[below] at (5, 0) {$5\tau$};
        \node[below] at (6, 0) {$6\tau$};

        \draw[ultra thick]
        (0.000, 7.000) -- (0.011, 7.000) --
        (0.011, 6.854) -- (0.013, 6.854) --
        (0.013, 6.708) -- (0.025, 6.708) --
        (0.025, 6.562) -- (0.050, 6.562) --
        (0.050, 6.417) -- (0.069, 6.417) --
        (0.069, 6.271) -- (0.078, 6.271) --
        (0.078, 6.125) -- (0.110, 6.125) --
        (0.110, 5.979) -- (0.127, 5.979) --
        (0.127, 5.833) -- (0.161, 5.833) --
        (0.161, 5.688) -- (0.176, 5.688) --
        (0.176, 5.542) -- (0.182, 5.542) --
        (0.182, 5.396) -- (0.184, 5.396) --
        (0.184, 5.250) -- (0.185, 5.250) --
        (0.185, 5.104) -- (0.210, 5.104) --
        (0.210, 4.958) -- (0.295, 4.958) --
        (0.295, 4.812) -- (0.324, 4.812) --
        (0.324, 4.667) -- (0.346, 4.667) --
        (0.346, 4.521) -- (0.357, 4.521) --
        (0.357, 4.375) -- (0.393, 4.375) --
        (0.393, 4.229) -- (0.458, 4.229) --
        (0.458, 4.083) -- (0.466, 4.083) --
        (0.466, 3.938) -- (0.651, 3.938) --
        (0.651, 3.792) -- (0.736, 3.792) --
        (0.736, 3.646) -- (0.814, 3.646) --
        (0.814, 3.500) -- (0.822, 3.500) --
        (0.822, 3.354) -- (0.951, 3.354) --
        (0.951, 3.208) -- (0.972, 3.208) --
        (0.972, 3.062) -- (0.988, 3.062) --
        (0.988, 2.917) -- (1.020, 2.917) --
        (1.020, 2.771) -- (1.163, 2.771) --
        (1.163, 2.625) -- (1.260, 2.625) --
        (1.260, 2.479) -- (1.266, 2.479) --
        (1.266, 2.333) -- (1.280, 2.333) --
        (1.280, 2.188) -- (1.281, 2.188) --
        (1.281, 2.042) -- (1.285, 2.042) --
        (1.285, 1.896) -- (1.653, 1.896) --
        (1.653, 1.750) -- (1.894, 1.750) --
        (1.894, 1.604) -- (2.000, 1.604) --
        (2.000, 1.458) -- (2.075, 1.458) --
        (2.075, 1.312) -- (2.165, 1.312) --
        (2.165, 1.167) -- (2.178, 1.167) --
        (2.178, 1.021) -- (2.285, 1.021) --
        (2.285, 0.875) -- (2.498, 0.875) --
        (2.498, 0.729) -- (2.556, 0.729) --
        (2.556, 0.583) -- (3.114, 0.583) --
        (3.114, 0.438) -- (3.455, 0.438) --
        (3.455, 0.292) -- (4.708, 0.292) --
        (4.708, 0.146) -- (5.223, 0.146) --
        (5.223, 0.000);
        \draw[->, cyan, ultra thick] (0.0 , 7.0) -- (0.95, 1.6);
        \draw[->, cyan, ultra thick] (1.05, 1.5) -- (1.95 , 0.6);
        \draw[->, cyan, ultra thick] (2.05, 0.56) -- (2.95 , 0.12);
        \draw[->, cyan, ultra thick] (3.05, 0.1) -- (3.95 , 0.05);
        \draw[->, cyan, ultra thick] (4.05, 0.0) -- (4.95 , 0.0);
        \draw[->, cyan, ultra thick] (5.05, 0.0) -- (5.95 , 0.0);
    \end{tikzpicture}
    \caption{
        An illustration of the number of $A$ atoms in a run of the Gillespie (black) and the $\tau$-leaping (cyan) algorithms.
    }\label{figure:tau_leap}
\end{figure}
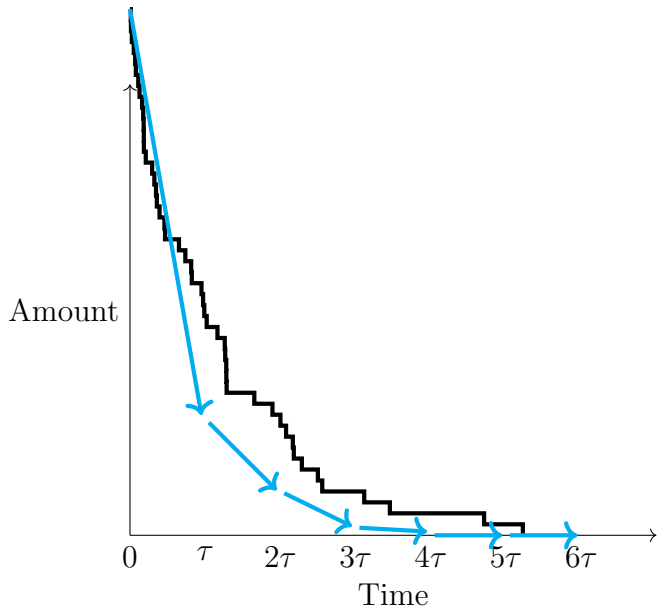

The $\tau$-leaping method is faster than the Gillespie algorithm, 
since every iteration can handle many radioactive decays.
However, it gives inaccurate results, 
since the assumption that the rate does not change along each time segment is incorrect.
In the worst case, since Poisson random variables are unbounded, 
we can get more radioactive decays than there are reactants, 
and the simulation could have a negative number of $A$ atoms remaining.

\paragraph{The $\tau$-splitting algorithm}

We present here a simplified description of the $\tau$-splitting algorithm.
For the algorithm to sample from the correct distribution, 
steps have to be correctly conditioned on previous steps, 
which we skip to simplify the presentation.
The complete description of the algorithm is given in 
Sections~\ref{section:basic} and~\ref{section:tau_improved}.

The $\tau$-splitting algorithm, similarly to the $\tau$-leaping method, 
simulates the reaction network by sampling the number of occurrences of each reaction during a time period
using a constant rate approximation.
For a reaction $R_i$ with a constant rate $\rho_i$, 
the number $e_i$ of occurrences of the reaction over a time period of length $t$
is distributed as $\Pois{\rho_i t}$.
Unlike the $\tau$-leaping method, the $\tau$-splitting algorithm also 
samples lower and upper bounds $l_i$ and $u_i$ for the rate,
such that the sample of $\Pois{l_i t}$ and $\Pois{u_i t}$ 
would still\footnote{
    To formalize this notion, we will define a coupling of the random variables $\Pois{\gamma}$ for all $\gamma \in \RR^+$,
    such that the sample of the Poisson random variables, viewed as a function from $\RR^+ \to \NN$, is a piecewise constant non-decreasing function.
    The values $l_i$ and $u_i$ are the lower and upper bounds of the segment in which $\rho_i$ is contained.
} be equal to $e_i$.
Intuitively, it makes sense that changing the reaction rate very 
little should not change the number of times it occurs, 
and the $\tau$-splitting algorithm samples the bounds $l_i$ and $u_i$ 
in a way that formalizes this intuition.

The $\tau$-splitting algorithm then checks how much the change in the state over 
the time period could change the reaction rates.
If for a reaction $R_i$, despite the change in the state over the time period, 
the reaction rate $\rho_i$ remains in $\clopen{l_i, u_i}$,
then the change in the reaction rate is guaranteed not to affect the number of occurrences of $R_i$ over the time period, 
and the constant-rate approximation is accurate.
If for all reactions $R_i$ the rate $\rho_i$ is guaranteed to remain in $\clopen{l_i, u_i}$, 
then the step is accurate, 
and the $\tau$-splitting algorithm can use it to to advance the simulation.
If for some reaction $R_i$, the rate $\rho_i$ is not guaranteed to remain in $\clopen{l_i, u_i}$ during the time period,
then the real number of occurrences of $R_i$ might not be equal to the $e_i$ sampled using the constant rate approximation.
In this case, the $\tau$-splitting algorithm splits the time period into two halves,
and recursively simulates each half.
It continues to split the time period this way until in every short time period the constant-rate approximation
is shown to be accurate.

We will now illustrate the $\tau$-splitting algorithm when simulating 
a single trajectory of the radioactive decay system (Table~\ref{table:tau_split}, Figure~\ref{figure:tau_split_decay}).
Since the example describes a single trajectory, 
all concrete values given are samples of random variables.
The trajectory simulated by the Gillespie algorithm is drawn in black, 
and the steps sampled by the $\tau$-splitting algorithm are drawn in orange and cyan.
Each row in the table repesents a step of the $\tau$-splitting algorithm.
Following the columns of Table~\ref{table:tau_split} from left to right, we check the number of $A$ atoms at the beginning of the step, compute the reaction rate $\rho$,
sample the number of occurrences $e$ and the lower bound $l$ of the rate such that 
the decay of the rate $\rho$ to $l$ 
will not change the number of occurrences.
We then compare the rate at the end of the step to $l$.
If it is larger than $l$, the step is declared accurate, and we progress the system.
Otherwise, we split the step into two halves.
Since in the example $\rho$ cannot increase, 
the check that $\rho < u$ is redundant, 
and we omit $u$ from the discussion.
Suppose that the initial number of $A$ atoms is 48 and that the reaction rate is $\gamma$.

\textbf{Row 1}: 
The algorithm tries to simulate $\clopen{0, T}$. The reaction rate at time 0 is $48\gamma$.
The algorithm samples the number of occurrences from $\Pois{48\gamma T}$, and gets $e=55$.
It then samples the lower bound $l$ and gets $26\gamma$.
Since the step would bring the reaction rate to $(48 - 55)\gamma = -7\gamma$, which is lower than $26\gamma$,
the step is rejected, and the algorithm thus splits $\clopen{0, T}$ to $\clopen{0, T/2}$ and $\clopen{T/2, T}$, 
and simulates the two periods recursively.

\textbf{Row 2}: 
The algorithm tries to simulate $\clopen{0, T/2}$. It samples $e = 47$ from $\Pois{48\gamma T/2}$, 
and rejects the step similarly to the first step,
making it simulate $\clopen{0, T/4}$ and $\clopen{T/4, T/2}$.

\textbf{Row 3}:
The algorithm tries to simulate $\clopen{0, T/4}$. It samples $e = 35$ from $\Pois{48\gamma T / 4}$
and $l = 9\gamma$.
The rate at the end of the step would be $(48 - 35)\gamma = 13\gamma$, which is greater than $l$,
so the step is accepted. The algorithm can thus proceed to simulate $\clopen{T/4, T/2}$.
It then simularly simulates $\clopen{T/4, T/2}$ (Row 4) and $\clopen{T/2, T}$ (Row 5).

\begin{table}[H]
    \centering
    \caption{Progress of the $\tau$-splitting algorithm}
    \label{table:tau_split}
    \begin{tabular}{l@{ }l@{ }lcccccc}
        \toprule
        \multicolumn{3}{c}{Time Segment} & $\num A$ & $\rho$ at step start & $e$ & $l$ & $\rho$ at step end & Accurate \\
        \midrule
        $[0, $ & $T$ & $)$ & 48 & $48\gamma$ & 55 & $26\gamma$ & $-7\gamma$  & \ding{55} \\
        $[0, $ & $T/2$ & $)$ & 48 & $48\gamma$ & 47 & $9\gamma$ & $\gamma$ & \ding{55} \\
        $[0, $ & $T/4$ & $)$ & 48 & $48\gamma$ & 35 & $9\gamma$ & $13\gamma$ & \ding{51} \\
        $[T/4, $ & $T/2$ & $)$ & 13 & $13\gamma$ & 9 & $2\gamma$ & $4\gamma$ & \ding{51} \\
        $[T/2, $ & $T$ & $)$ & 4 & $4\gamma$ & 4 & 0 & 0 & \ding{51} \\
        \bottomrule
    \end{tabular}
\end{table}

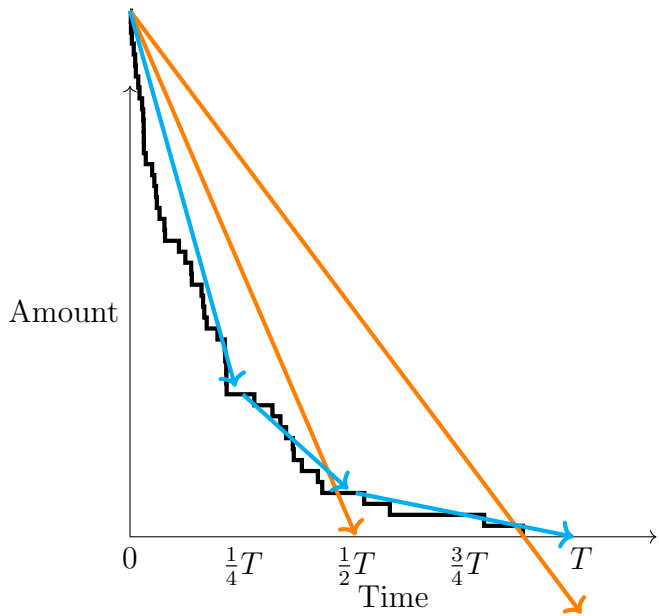
\begin{figure}[H]
    \centering
    
    \begin{tikzpicture}
        \draw[->] (0,0) -- (7,0); 
        \draw[->] (0,0) -- (0,6); 
        
        \node[below] at (0, 0) {0};
        \node[below] at (1.5, 0) {$\frac{1}{4}T$};
        \node[below] at (3, 0) {$\frac{1}{2}T$};
        \node[below] at (4.5, 0) {$\frac{3}{4}T$};
        \node[below] at (6, 0) {$T$};
        \node[below] at (3.5, -0.5) {Time};
        \node[left] at (0, 3) {Amount};

        \draw[ultra thick]
        (0.000, 7.000) -- (0.011, 7.000) --
        (0.011, 6.854) -- (0.013, 6.854) --
        (0.013, 6.708) -- (0.025, 6.708) --
        (0.025, 6.562) -- (0.050, 6.562) --
        (0.050, 6.417) -- (0.069, 6.417) --
        (0.069, 6.271) -- (0.078, 6.271) --
        (0.078, 6.125) -- (0.110, 6.125) --
        (0.110, 5.979) -- (0.127, 5.979) --
        (0.127, 5.833) -- (0.161, 5.833) --
        (0.161, 5.688) -- (0.176, 5.688) --
        (0.176, 5.542) -- (0.182, 5.542) --
        (0.182, 5.396) -- (0.184, 5.396) --
        (0.184, 5.250) -- (0.185, 5.250) --
        (0.185, 5.104) -- (0.210, 5.104) --
        (0.210, 4.958) -- (0.295, 4.958) --
        (0.295, 4.812) -- (0.324, 4.812) --
        (0.324, 4.667) -- (0.346, 4.667) --
        (0.346, 4.521) -- (0.357, 4.521) --
        (0.357, 4.375) -- (0.393, 4.375) --
        (0.393, 4.229) -- (0.458, 4.229) --
        (0.458, 4.083) -- (0.466, 4.083) --
        (0.466, 3.938) -- (0.651, 3.938) --
        (0.651, 3.792) -- (0.736, 3.792) --
        (0.736, 3.646) -- (0.814, 3.646) --
        (0.814, 3.500) -- (0.822, 3.500) --
        (0.822, 3.354) -- (0.951, 3.354) --
        (0.951, 3.208) -- (0.972, 3.208) --
        (0.972, 3.062) -- (0.988, 3.062) --
        (0.988, 2.917) -- (1.020, 2.917) --
        (1.020, 2.771) -- (1.163, 2.771) --
        (1.163, 2.625) -- (1.260, 2.625) --
        (1.260, 2.479) -- (1.266, 2.479) --
        (1.266, 2.333) -- (1.280, 2.333) --
        (1.280, 2.188) -- (1.281, 2.188) --
        (1.281, 2.042) -- (1.285, 2.042) --
        (1.285, 1.896) -- (1.653, 1.896) --
        (1.653, 1.750) -- (1.894, 1.750) --
        (1.894, 1.604) -- (2.000, 1.604) --
        (2.000, 1.458) -- (2.075, 1.458) --
        (2.075, 1.312) -- (2.165, 1.312) --
        (2.165, 1.167) -- (2.178, 1.167) --
        (2.178, 1.021) -- (2.285, 1.021) --
        (2.285, 0.875) -- (2.498, 0.875) --
        (2.498, 0.729) -- (2.556, 0.729) --
        (2.556, 0.583) -- (3.114, 0.583) --
        (3.114, 0.438) -- (3.455, 0.438) --
        (3.455, 0.292) -- (4.708, 0.292) --
        (4.708, 0.146) -- (5.223, 0.146) --
        (5.223, 0.000);
        \draw[->, orange, ultra thick] (0.0 , 7.0) -- (6.0, -1.02);
        \draw[->, orange, ultra thick] (0.0 , 7.0) -- (3.0, 0.02);
        \draw[->, cyan, ultra thick] (0.0 , 7.0) -- (1.4, 1.996);
        \draw[->, cyan, ultra thick] (1.5, 1.896) -- (2.9 , 0.633);
        \draw[->, cyan, ultra thick] (3.0,  0.583) -- (5.9 , 0.0);
    \end{tikzpicture}
    \caption{
        An illustration of the $\tau$-splitting algorithm.
        The trajectory sampled by the Gillespie algorithm is drawn in black.
        The algorithm first attempts to step directly to time $T$ (orange arrow pointing at $(T, -7)$).
        The rate changes significantly in that time period, and the step is rejected.
        The algorithm then attempts to step to $\frac{1}{2}T$ (orange arrow pointing at $(\frac{1}{2}T, 1)$).
        The rate changes significantly in that step as well, and the step is rejected.
        The algorithm next attempts to step to $\frac{1}{4}T$ (cyan arrow).
        The rate does not change significantly in that time period, and the step is accepted.
        The next two steps over the time periods $\clopen{\frac 1 4 T, \frac 1 2 T}$ and $\clopen{\frac 1 2 T, T}$ 
        are accepted as well (cyan arrows).
    }\label{figure:tau_split_decay}
\end{figure}

In a system with multiple reactions, 
the algorithm searches for leaps for each reaction independently.

\subsection{Problem statement}

We will now formally define the problem we are to solve.
While we use continuous time markov chains (CTMCs) defined by the CME~\cite{gillespie1992rigorous} 
formulation, the algorithm can be easily adapted to any setting where the  
transition rates depend on the state in a sufficiently controlled fashion.
The algorithm can be extended to incorporate variations of the CME model, 
such as complex rate functions~\cite{thanh2020rssalib}, 
time-dependent rates~\cite{anderson2007time_dependent}, 
and delayed reactions~\cite{cai2007delayed}.

For a vector $v$, let $v[i]$ denote the $i$'th component of the vector.

Let the state of the system be $S \in \NN^s$, 
where $\NN$ is the set of nonnegative integers and $s$ is the dimension of the state space, 
and let $S_t$ be the state of the system at time $t$.
The quantity $S[k]$ is interpreted as the amount of the $k$'th \emph{reactant} in the chemical system.
The transitions of the system are described by a set of $d$ chemical reactions $\cR$.
A \emph{reaction} $R_i = (I_i, C_i, O_i, r_i)$ is composed of the set of inputs $I_i$, 
the input count vector $C_i$, 
the stoichiometry vector $O_i$,
and a base reaction rate $r_i$.
For example, in a system with four reactants $C_1, C_2, C_3, C_4$,
the reaction
\begin{align}
    C_1 + C_2 \to 2C_1 + C_3
\end{align}
would have the input set $I = \cbr{1, 2}$, 
the input count vector $C = \pr{1, 1, 0, 0}$, 
and the stoichiometry vector $O = \pr{1, -1, 1, 0}$.
In the systems we consider, the input count vector and the stoichiometry vector are usually sparse, 
having $O(1)$ nonzero components.

The rate at which the reaction $R_i$ occurs is $\rho_i(S) = r_i \prod_{k \in I_i} \binom{S[k]}{C_i[k]}$,
which is equal to the base rate of the reaction times the number of combinations of reactants in which the $k$'th reactant appears $C_i[k]$ times.
Assuming no interference from other reactions, 
the time until a reaction $R_i$ occurs is distributed $\Exp{\rho_i(S)}$, 
where $\mathrm{Exp}$ denotes the exponential distribution. 
When the reaction occurs, it changes the state from $S$ to $S + O_i$.

Since the rate at which a reaction $R_i$ occurs at state $S$
is $\rho_i(S)$, the total rate at which the state changes at $S$
is $\sum_i \rho_i(S)$, and the time until the state change is distributed $\Exp{\sum_i \rho_i(S)}$.
When the state changes, with probability $\frac{\rho_i(S)}{\sum_j \rho_j(S)}$ the reaction $R_i$ occurs,
changing the state from $S$ to $S + O_i$.
Given the initial state, the distribution $\PP_t$ of $S_t$, 
the state at time $t$, 
is determined by the transitions.
Given a time $T > 0$, the goal is to sample the state at time $T$.
We present the $\tau$-splitting algorithm, which solves this problem.
Unlike the Gillespie algorithm, that samples the whole trajectory of the system, 
the $\tau$-splitting algorithm samples only a subset of the trajectory,
such that in each time segment the changes in the state are distributed uniformly and independently.

\subsection{The main result}

We have implemented an optimized version of the full $\tau$-splitting algorithm in the Rust programming language,
that outperforms previous state-of-the-art algorithms,
and proved that it has an improved runtime complexity.
We compare the time it takes for the algorithm to simulate the 
Fc$\epsilon$RI signalling network~\cite{liu2013fceri} and the 
B-cell receptor signalling network~\cite{barua2012computational}.
The Fc$\epsilon$RI signalling network has 380 reactants and 3862 reactions.
The B-cell receptor signalling network formally has 1122 reactants and 24388 reactions, 
but 108 reactions are given with a rate of 0 and were removed by us from the system, 
leaving 24280 reactions.
We use the initial states provided by~\cite{ghosh2021blsssa}.

We compare our algorithm to the RSSA-CR implementation in~\cite{thanh2020rssalib} and the BlSSSA implementation in~\cite{ghosh2021blsssa}.
Both our algorithm and the RSSA-CR implementation simulate time periods, 
while BlSSSA simulates reaction counts. 
To compare them, we applied the $\tau$-splitting algorithm and the RSSA-CR algorithm on a time period 
where there are at least $10^7$ reaction events,
and simulated exactly $10^7$ events using the BlSSSA algorithm.
We ran each algorithm five times on an Intel i5-1035G1 CPU and averaged the results.

Figure~\ref{figure:benchmarking} shows the runtime of the three 
algorithms when simulating $10^7$ reactions in the two systems.
For each of the two systems appear three bars, one for each algorithm.
The $\tau$-splitting algorithm is 4.8 times faster than the BlSSSA algorithm 
and 11.7 times faster than the RSSA-CR algorithm for the simulation
of the B-cell receptor signalling network, and 2.0 times faster than the BlSSSA algorithm
and 4.5 times faster than the RSSA-CR algorithm for the simulation of the Fc$\epsilon$RI signalling network.

\begin{figure}[H]
\begin{tikzpicture}
\begin{axis}[
    ybar,
    bar width=12pt,
    width=12cm,
    height=6cm,
    enlarge x limits=0.5,
    legend style={at={(0.5,-0.15)}, anchor=north, legend columns=-1},
    ylabel={Time (s)},
    symbolic x coords={B-Cell, Fc$\epsilon$RI},
    xtick=data,
    ymin=0,
    ymax=9,
    nodes near coords,
    nodes near coords align={vertical},
    bar shift=0pt,
    axis x line*=bottom,
    axis y line*=left,
    major x tick style = transparent,
]

\addplot+[ybar, bar shift=-18pt, fill=cyan, draw=none] 
    coordinates {(B-Cell,7.73) (Fc$\epsilon$RI,6.06) };  
\addplot+[ybar, bar shift=0pt, fill=RedOrange, draw=none] 
    coordinates {(B-Cell,3.19) (Fc$\epsilon$RI,2.73)};
\addplot+[ybar, bar shift=18pt, fill=ForestGreen, draw=none] 
    coordinates {(B-Cell,0.66) (Fc$\epsilon$RI,1.33)};  

\legend{RSSA-CR, BlSSSA, $\tau$-splitting}

\end{axis}
\end{tikzpicture}
\caption{
    The average runtime for the simulation of $10^7$ reations 
    by the RSSA-CR, the BlSSSA, and the $\tau$-splitting algorithms.
}\label{figure:benchmarking}
\end{figure}
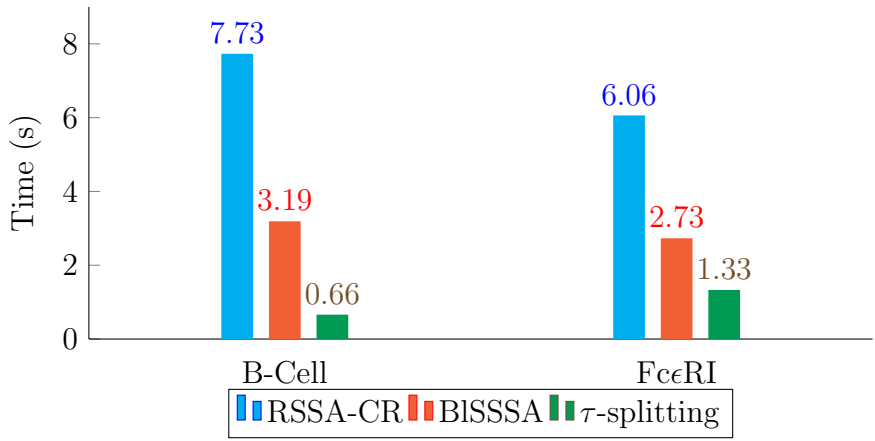

A reaction is called \emph{stable} when the change in its rate over a time 
period does not change the number of occurrences.
When this happens, the recursion of the $\tau$-splitting algroithm for that reaction ends,
and the algorithm spends no more runtime for that reaction. 
By stopping processing a reaction that has several occurrences, 
the algorithm decreases its runtime compared to the Gillepsie algorithm.

The mean number of occurrences in a reaction when it becomes stable is 16.9 
in the B-Cell receptor system and 11.86 in the Fc$\epsilon$RI system.
We plot the histogram of the number of occurrences in reactions at the step where they became stable
in the B-Cell receptor signalling network in Figure~\ref{figure:step_size_histogram}.
While the most common number of occurrences is 0, 
there is a significant number of steps with a larger number of occurrences,
which the Gillespie algorithm spends many steps simulating.

\begin{figure}[H]
\begin{tikzpicture}
\begin{axis}[
    xlabel={Reaction occurence count},
    ylabel={Amount},
    ymode=log,
    smooth,
    grid=major,
    width=14cm,
    height=8cm
]
\addplot+[blue, thick, mark=none] coordinates {
(0, 57350) (1, 43866) (2, 33340) (3, 27659) (4, 24231) (5, 22498) (6, 20979) (7, 19475) (8, 18414) (9, 17276) (10, 16315) (11, 15277) (12, 14903) (13, 14247) (14, 13486) (15, 13071) (16, 12293) (17, 11687) (18, 11177) (19, 10642) (20, 10067) (21, 9290) (22, 8914) (23, 8253) (24, 7711) (25, 7398) (26, 7094) (27, 6919) (28, 6611) (29, 6355) (30, 6197) (31, 6030) (32, 5837) (33, 5542) (34, 5189) (35, 5009) (36, 4655) (37, 4372) (38, 4158) (39, 3731) (40, 3471) (41, 3299) (42, 3017) (43, 2939) (44, 2725) (45, 2550) (46, 2308) (47, 2150) (48, 1971) (49, 1925) (50, 1849) (51, 1708) (52, 1631) (53, 1534) (54, 1471) (55, 1485) (56, 1478) (57, 1341) (58, 1300) (59, 1235) (60, 1191) (61, 1116) (62, 1070) (63, 1107) (64, 1035) (65, 955) (66, 875) (67, 877) (68, 828) (69, 733) (70, 727) (71, 671) (72, 644) (73, 645) (74, 615) (75, 564) (76, 520) (77, 475) (78, 463) (79, 438) (80, 410) (81, 372) (82, 354) (83, 303) (84, 309) (85, 305) (86, 286) (87, 256) (88, 203) (89, 196) (90, 167) (91, 181) (92, 160) (93, 162) (94, 156) (95, 127) (96, 138) (97, 136) (98, 143) (99, 137) (100, 123) (101, 128) (102, 106) (103, 101) (104, 115) (105, 108) (106, 113) (107, 95) (108, 88) (109, 95) (110, 83) (111, 78) (112, 87) (113, 76) (114, 63) (115, 64) (116, 65) (117, 62) (118, 56) (119, 58) (120, 54) (121, 67) (122, 42) (123, 59) (124, 66) (125, 58) (126, 30) (127, 48) (128, 35) (129, 36) (130, 35) (131, 32) (132, 38) (133, 30) (134, 35) (135, 26) (136, 25) (137, 18) (138, 28) (139, 29) (140, 21) (141, 17) (142, 16) (143, 28) (144, 20) (145, 17) (146, 16) (147, 13) (148, 12) (149, 12) (150, 18) (151, 14) (152, 14) (153, 16) (154, 11) (155, 11) (156, 7) (157, 9) (158, 14) (159, 22) (160, 12) (161, 8) (162, 8) (163, 6) (164, 11) (165, 6) (166, 6) (167, 7) (168, 8) (169, 8) (170, 10) (171, 2) (172, 7) (173, 4) (174, 2) (175, 4) (176, 6) (177, 3) (178, 7) (179, 4) (180, 6) (181, 4) (182, 4) (183, 7) (184, 2) (185, 2) (186, 2) (187, 3) (188, 3) (189, 1) (190, 3) (192, 1) (193, 1) (194, 4) (195, 4) (196, 2) (197, 3) (198, 5) (199, 1) (201, 1) (202, 2) (203, 3) (204, 2) (206, 1) (207, 4) (208, 1) (209, 1) (210, 1) (211, 1) (212, 1) (213, 1) (216, 1) (219, 1) (220, 1) (222, 1) (226, 1) (227, 1) (230, 1) (231, 1) (233, 1) (236, 3) (241, 1) (251, 1) (256, 1) (260, 1) (261, 1) (266, 2) (268, 3) (269, 2) (271, 2) (272, 2) (273, 1) (274, 3) (276, 1) (277, 1) (278, 1) (279, 1) (283, 2) (284, 1) (285, 1) (286, 4) (287, 3) (288, 4) (289, 1) (290, 6) (291, 4) (292, 2) (294, 5) (295, 1) (296, 4) (297, 7) (298, 6) (299, 2) (300, 2) (301, 6) (302, 3) (303, 5) (304, 4) (305, 4) (306, 4) (307, 4) (308, 5) (309, 3) (310, 1) (311, 4) (312, 3) (313, 2) (314, 2) (315, 1) (316, 3) (317, 1) (318, 3) (319, 1) (321, 1) (322, 2) (323, 3) (325, 1) (326, 2) (327, 2) (328, 2) (331, 1) (332, 3) (333, 1) (335, 1) (341, 1) (346, 1)
};
\end{axis}
\end{tikzpicture}
\caption{
    A histogram of the number of occurrences in reactions at the step where $\rho$ 
    becomes bounded between $l$ and $u$.
    A small fraction of the reactions that had more than 350 occurrences 
    at that point are removed from the histogram for visual clarity.
}\label{figure:step_size_histogram}

\end{figure}
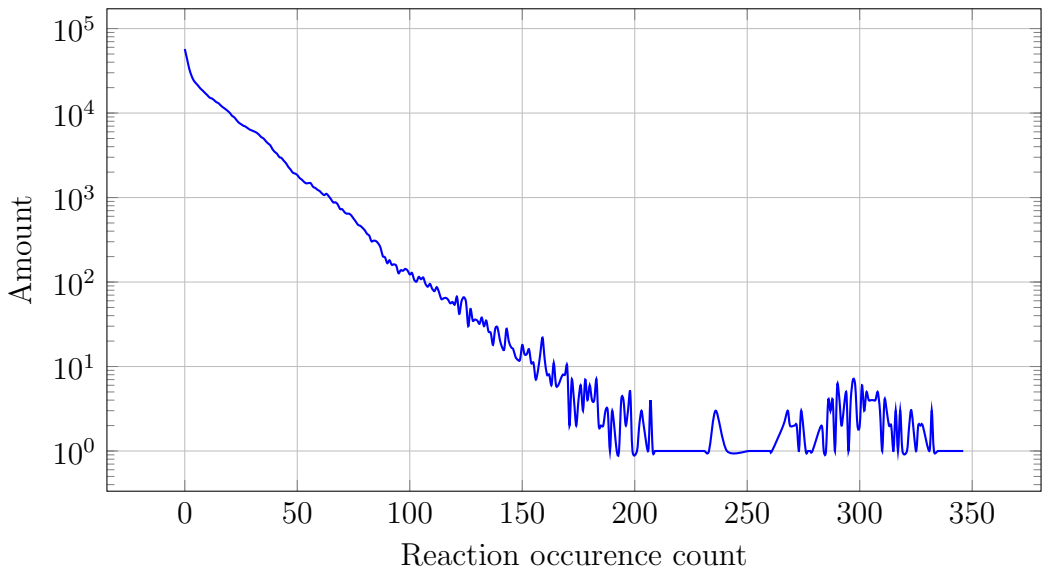

We will now present the complexity of the $\tau$-splitting algorithm.
The results here will be stated in their simplest form, 
and their technical form will be stated and proved in Section~\ref{section:full:complexity}.

The runtime of the $\tau$-splitting algorithm is a random variable, 
and it would be natural to estimate its expectation.
However, the runtime is highly dependent on the trajectory, 
and the distribution of the trajectories is difficult to predict from the algorithm's input,
making the estimation difficult.
Instead, we will define the \emph{runtime density} to be a function from the set of states to $\RR$,
such that the expected runtime of the algorithm is the integral of the runtime density over all trajectories.
To describe the complexity, we will also need the concept of reactions \emph{depending} on other reactions.
A reaction $R_i$ \emph{depends} on a reaction $R_j$ if the occurrence of $R_j$ can change the rate of $R_i$, 
which happens when for some $k \in I_i$ we have $O_j[k] \ne 0$.

To illustrate the concept of runtime density, 
we will compute it for two optimized versions of the Gillespie algorithm, 
the RSSA-CR algorithm~\cite{thanh2016rssacr} 
and the BlSSSA algorithm~\cite{ghosh2021blsssa}.
\begin{lemma}
    Let $D_i'$ be the number of reactions depending on $R_i$.
    The runtime density of the RSSA-CR algorithm
    is $\bigo{\sum_i (D_i' + K)\rho_i(S)}$,
    where $K$ is the number of groups used by the algorithm,
    and the runtime density of the BlSSSA algorithm
    is $\bigo{\sum_i \sqrt{s}\rho_i(S)}$.
\end{lemma}

The complexity of the $\tau$-splitting algorithm, stated informally, is:
\begin{theorem}
    Let $D_i$ be the total number of reactions depending on $R_i$ or on which $R_i$ depends.
    Let $\rho_k^\pm = \sum_j \rho_j(S) \abs{O_j[k]}$.
    
    When the values $S[k]$ are sufficiently large, the function
    \begin{align*}
        f(S) = \bigo{\sum_i D_i \sqrt{\rho_i(S) \sum_{k \in I_i} \frac{\rho_k^\pm}{S[k]}}}
    \end{align*}
    is a runtime density for the $\tau$-splitting algorithm.
\end{theorem}

We will now compare the runtime densities of the 
RSSA-CR, BlSSSA, and $\tau$-splitting algorithms.
All three complexities have a term coming from the frequency of the events they process
and a term for the cost of processing the event.

The BlSSSA and the RSSA-CR algorithms process each reaction,
and the rate of reactions needing to be processed is $\sum_i \rho_i(S)$.
The $\tau$-splitting algorithm only spends time when reaction occurrences 
significantly change the reaction rate, the rate of which is only
$\sqrt{\rho_i(S) \sum_{k \in I_i} \frac{\rho_k^\pm}{S[k]}}$ 
(This is proved in Section~\ref{section:full:complexity}).
The larger the values of $S[k]$ are, the lower the rate of events needing to be processed 
by the $\tau$-splitting algorithm compared to the other two algorithms.

The update cost terms are all different.
The update cost for the $\tau$-splitting algorithm is $D_i$, 
the number of reactions depending on the $i$'th reaction or on which the $i$'th reaction depends.
The update cost for the RSSA-CR algorithm is $D_i' + K$.
The term $D_i'$ is the number of reactions depending on the $i$'th reaction,
and the term $K$ is an implementation detail of the RSSA-CR algorithm, which is generally smaller than $D_i'$ and we will not discuss.
The average of $D_i'$ is half the average of $D_i$ from graph-theoretic considerations,
but the difference in their contribution to the runtime could be significant for 
some graph structures.
The term $\sqrt s$ of the BlSSSA algorithm is significantly different from $D_i$ and $D_i'$.
When most reactions affect only a few other reactions, $\sqrt s > D_i > D_i'$.
When the reactions are highly dependent on each other, $D_i > D_i' > \sqrt s$.

To summarize, the $\tau$-splitting algorithm dominates on loosely coupled networks 
where the values $S[k]$ are large,
the BlSSSA algorithm dominates on densely coupled reaction networks with low $S[k]$,
and the RSSA-CR algorithm dominates in loosely coupuld networks with low $S[k]$.

An important research question is finding an algorithm that combines all advantages.
Furthermore, the $\tau$-splitting algorithm uses a first-order approximation. 
Using higher-order approximations could lead to further runtime improvements.

\section{The basic \texorpdfstring{$\tau$}{tau}-splitting algorithm}\label{section:basic}

We here describe a basic version of the $\tau$-splitting algorithm.
In Section~\ref{section:tau_improved} we will present the full version of the $\tau$-splitting algorithm, which has a better time complexity.
The basic version is presented for pedagogical reasons and to aid in proving the algorithm's correctness. 

\subsection{\texorpdfstring{$\tau$}{tau}-leaping steps}\label{section:basic:tau_steps}

Recall from Section~\ref{section:informal} that to simulate a period $\timeseg{0}{1}$,
the $\tau$-splitting algorithm samples a step using the constant rate approximation, 
and samples the change in the reaction rate that would make the step inaccurate.
We begin by describing the data we sample for a step.

The basic data structure used by the algorithm is the \emph{$\tau$-leaping step}, 
which extends the steps of the $\tau$-leaping method~\cite{gillespie2001approximate}.
The $\tau$-leaping step describes the progress of the system from the time $t_0$ to the time $t_1$.
It consists of:
\begin{itemize}
    \item An initial state $S \in \NN^s$. This is the state used to compute the reaction rates, 
    and is usually equal to the state $S_{t_0}$.
    \item A vector\footnote{Recall that $d$ is the number of reactions.} $\pr{e_i}_i \in \NN^d$. 
        For each $i$, $e_i$ is the putative number of occurrences of the reaction $R_i$ 
        between $t_0$ and $t_1$, assuming a constant reaction rate.
    \item The vectors $\pr{l_i}_i \in \RR^d$ and $\pr{u_i}_i \in \RR^d$. 
        For each $i$, the values $l_i$ and $u_i$ are the lower and upper bounds on the rate of $R_i$ 
        where the algorithm would still sample the same $e_i$.
        This notion will be formally defined in Section~\ref{section:basic:tau_sampling}.
\end{itemize}
The algorithm also uses $\tau$-leaping steps where $S \ne S_{t_0}$ as intermediates.
We will call such improper steps \emph{$\tau$-leaping precursors}.

The $\tau$ splitting algorithm works by sampling an initial 
$\tau$-leaping step for the time segment $\clopen{0, T}$, 
then recursively splitting $\tau$-leaping steps to improve their accuracy.
The algorithm uses the three building blocks: 
The stopping condition (Procedure~\ref{procedure:stop_condition}), 
splitting $\tau$-leaping steps (Procedure~\ref{procedure:tau:split}), 
and resampling $\tau$-leaping steps (Procedure~\ref{procedure:tau:resample}). 
The three procedures will be elaborated on in Sections~\ref{section:basic:tau_sampling} and~\ref{section:basic:stop_condition}.

The first building block is the \emph{stopping condition}, 
that determines when a step is accurate and when it should be split further.
We formally define a stopping condition the following way:
\begin{procedure}[Stopping condition]\label{procedure:stop_condition}
    A stopping condition is a function from $\tau$-leaping steps to the set $\cbr{\text{STOP}, \text{SPLIT}}$.
\end{procedure}
We will say that a step \emph{satisfies} a stopping condition if the step is mapped by the stopping condition to STOP\@.


We will define the stopping condition used in the $\tau$-splitting algorithm in Section~\ref{section:basic:stop_condition}.

The second building block is the procedure used to split $\tau$-leaping steps into two halves.
Given two times $t_0 < t_1$, denote $t_{1/2} = (t_0 + t_1)/2$.

\begin{procedure}[Splitting $\tau$-leaping steps]\label{procedure:tau:split}
    Given a $\tau$-leaping step over the time segment $\timeseg{0}{1}$, 
    splits the $\tau$-leaping step into a 
    $\tau$-leaping step over the time segment $\timeseg{0}{1/2}$ 
    and a $\tau$-leaping precursor over $\timeseg{1/2}{1}$ 
    with the initial state $S_{t_0}$ (rather than $S_{t_{1/2}}$).
\end{procedure}

Procedure~\ref{procedure:tau:split} produces a $\tau$-leaping step for the time segment $\timeseg{1/2}{1}$,
where the values of $e_i$ were sampled using the rates at time $t_0$ and not the rates at 
time $t_{1/2}$. 
The third building block resamples the values of $e_i$ according to the change of rates from $t_0$ to $t_{1/2}$.

\begin{procedure}[Resampling $\tau$-leaping steps]\label{procedure:tau:resample}
    Given a $\tau$-leaping precursor over the time segment $\timeseg{0}{1}$ with
    an initial state $S$, and the state $S_{t_0}$,
    samples a $\tau$-leaping step over the same time segment with the initial state $S_{t_0}$.
\end{procedure}

For the algorithm to be correct, we need the following properties, which will be
elaborated on and proved in Sections~\ref{section:basic:tau_sampling} and~\ref{section:basic:stop_condition}.

The first property guarantees that when the stopping condition maps a $\tau$-leaping step to STOP,
then we can progress the system through the step's time segment:

\begin{property}[Exact stopping condition]\label{property:stop:exact}
    If a $\tau$-leaping step over $\timeseg{0}{1}$ with vectors $\pr{\pr{e_i}_i, \pr{l_i}_i, \pr{u_i}_i}$ 
    is mapped to STOP, then $S_{t_1} = S_{t_0} + \sum_i e_i O_i$.
\end{property}

The second property guarantees that all $\tau$-leaping steps share the same family of distributions,
regardless of the way they are sampled:

\begin{property}[Correct distribution]\label{property:tau:distribution}
    There exists a family of distributions $\cY(\rho, t)$ over $\ZZ\x\RR\x\RR$ 
    parametrized by the rate $\rho$ and the time $t$,
    such that in a $\tau$-leaping step over the time segment $\timeseg{0}{1}$,
    for all $i$ the tuple $\elh{_i}$ is distributed according to $\cY(\rho_i(S_{t_0}), t_1 - t_0)$.
    Moreover, the marginal distribution of $e_i$ is $\Pois{\rho_i (t_1 - t_0)}$, 
    the Poisson distribution with rate $\rho_i (t_1 - t_0)$.
\end{property}

The transitions of the discrete-state SDE depend only on the current state,
and not on the past trajectory. 
The third property guarantees that the way we sample $\tau$-leaping steps satisfies the same independence:

\begin{property}[Independence]\label{property:tau:independence}
    A $\tau$-leaping step over $\timeseg{0}{1}$ 
    is independently distributed of any $\tau$ leaping step ending before $t_0$ when conditioned on $S_{t_0}$,
    and all tuples $\elh{_i}$ in the step are independent from each other.
\end{property}

All $\tau$ leaping steps and precursors sampled by the algorithm will satisfy these properties.

\subsection{The basic \texorpdfstring{$\tau$}{tau}-splitting algorithm}\label{section:basic_algorithm}

We will now use the procedures defined in Section~\ref{section:basic:tau_steps}
to define the basic $\tau$-splitting algorithm.

To simulate the state of the system at time $T$ given an initial state $S_0$, 
the algorithm starts by sampling a $\tau$-leaping step for the time segment $\clopen{0, T}$.
Since $\tau$-leaping steps over long time segments may be inaccurate, the algorithm recursively
splits inaccurate $\tau$-leaping steps to increase their accuracy.

The initial step has initial state $S_0$, 
and every tuple $\elh{_i}$ is sampled from $\cY(\rho_i(S_0), T)$.
The recursive processing of a $\tau$-leaping step $\tau_{\timeseg{0}{1}}$ for the time segment $\timeseg{0}{1}$ has four steps:

\paragraph{Step 1.}\label{algorithm:basic:sample}
If $\tau_{\timeseg{0}{1}}$ satisfies the stopping condition, 
set $S_{t_1} = S_{t_0} + \sum_i e_i O_i$ and finish the recursive call.
Otherwise, set $t_{1/2} = (t_0 + t_1) / 2$, and split $\tau_{\timeseg{0}{1}}$ to 
a $\tau$-leaping step $\tau_{\timeseg{0}{1/2}}$ over $\timeseg{0}{1/2}$ 
and a $\tau$-leaping precursor $\tau_{\timeseg{1/2}{1}}^*$ over $\timeseg{1/2}{1}$ 
using Procedure~\ref{procedure:tau:split}.

\paragraph{Step 2.}\label{algorithm:basic:recurse_left}
Recursively apply the algorithm to $\tau_{\timeseg{0}{1/2}}$. After this step, the state $S_{t_{1/2}}$ is determined.

\paragraph{Step 3.}\label{algorithm:basic:resample}
Now that the state $S_{t_{1/2}}$ is determined, 
sample the $\tau$-leaping step $\tau_{\timeseg{1/2}{1}}$ conditioned on $\tau_{\timeseg{1/2}{1}}^*$ 
and the state $S_{t_{1/2}}$ using Procedure~\ref{procedure:tau:resample}.

\paragraph{Step 4.}\label{algorithm:basic:recurse_right}
Recursively apply the algorithm to $\tau_{\timeseg{1/2}{1}}$. After this step, the state $S_{t_1}$ is determined.

This concludes the recursive procedure of the $\tau$-splitting algorithm.
Our main result asserts that the algorithm samples $S_T$ according to $\PP_T$.

\begin{theorem}\label{theorem:exact_algorithm_correctness}
    Given an initial state $S_0$, a set of reactions $\cR$, and a time $T$,
    such that all for all $t, k$ the values $S_t[k]$ are bounded in probability,
    the $\tau$-splitting algorithm samples a state distributed as $\PP_T$.
\end{theorem}

The requirement that $S_t[k]$ be bounded in probability is necessary, 
since without it the state of the system could diverge to infinity in finite time,
which would prevent its simulation by either the $\tau$-splitting algorithm 
or the Gillespie algorithm.
This can happen when using reactions of the form $2A \to 3A$, 
but does not occur in reaction systems where rules such as conservation of mass are in effect.

\subsection{The sampling of \texorpdfstring{$\tau$}{tau}-leaping steps}\label{section:basic:tau_sampling}

We now turn to the sampling and splitting of $\tau$-leaping steps.
We first define the distribution $\cY$ used in Property~\ref{property:tau:distribution}, 
a function $f$ used to sample values from that distribution, 
and the vectors $\pr{l_i}_i$ and $\pr{u_i}_i$.
We will then show how to implement Procedures~\ref{procedure:tau:split} and~\ref{procedure:tau:resample} 
such that Properties~\ref{property:tau:distribution} and~\ref{property:tau:independence} are satisfied.

In the following subsections we will use the $\beta$ distribution to sample the 
highest and lowest of $n$ uniformly distributed points.
It will be convenient to define $\beta(1, 0)$ to be the delta function at 1,
and $\beta(0, 1)$ to be the delta function at 0.

\subsubsection{Definition of \texorpdfstring{$f$}{f} and \texorpdfstring{$\cY$}{Y}}
To sample $\pr{e_i}_i$ and define $\pr{l_i}_i$ and $\pr{u_i}_i$ 
we will use a Poisson point process\footnote{
    The Poisson point process on $\RR^2$ with rate $\rho$ samples points on the plane 
    such that the number of points in a subset with area $A$ 
    is a Poisson distributed random variable with rate $\rho A$, 
    and the number of points in disjoint subsets of the plane 
    are independent random variables~\cite{kingman1992poisson}.
} $\cX$ on $\RR^2$ with rate 1.
We consider a plane $\RR^2$, where the x axis corresponds to time 
and the y axis corresponds to the reaction rate (Figure~\ref{figure:tau_sampling}).

Let $X$ be a sample of the Poisson point process.
Let $f(X, \timeseg{0}{1}\x\clopen{0, \rho})$ be the following function 
from samples of the Poisson point process to $\ZZ \x \RR \x \RR$.
The image of $f$ is a tuple $(e, l, u)$, 
where $e$ is defined to be the number of points 
sampled by the Poisson point process inside the rectangle $\timeseg{0}{1}\x\clopen{0, \rho}$, 
$l$ is defined to be the height of the highest point in $\timeseg{0}{1}\x\clopen{0, \rho}$, 
and $u$ is defined to be the height of the lowest point in $\timeseg{0}{1}\x\clopen{\rho, \infty}$.
The distribution $\cY(\rho, t)$ is defined to be
the distribution of $f(\cX, \clopen{0, t} \x \clopen{0, \rho})$.

\begin{figure}[H]
\begin{tikzpicture}
\draw[->, thick] (0,0) -- (10,0);
\draw[->, thick] (0,0) -- (0,6);

\def\tzero{3};
\def\tone{7};
\def\rate{3};
\fill[red!20] (\tzero,0) rectangle (\tone,\rate);


\node[below] at (0, 0) {0};
\node[below] at (\tzero, 0) {$t_0$};
\node[below] at (\tone, 0) {$t_1$};
\node[below] at (9,0) {time};

\node[left] at (0,0) {0};
\node[left] at (0,3) {$\rho$};
\node[left] at (0,5) {reaction rate};

\draw[black] (\tzero, 0) -- (\tzero, 6);
\draw[black] (\tone, 0) -- (\tone, 6);
\draw[black] (0, \rate) -- (10, \rate);
\draw[blue, dashed] (0,2.4) -- (10,2.4);
\node[blue, left] at (0,2.4) {$l$};
\draw[red, dashed] (0,3.7) -- (10,3.7);
\node[red, left] at (0,3.7) {$u$};

\fill (1,1) circle (2pt);
\fill (2,2.1) circle (2pt);
\fill (3.5,1) circle (2pt); 
\fill (4,2.4) circle (2pt); 
\fill (6.5,1.1) circle (2pt); 
\fill (6,3.7) circle (2pt); 
\fill (4.5,4.5) circle (2pt);
\fill (8,3.2) circle (2pt);
\fill (6.8,5.8) circle (2pt);

\end{tikzpicture}

\caption{An illustration of the Poisson point process. The rectangle $\clopen{t_0, t_1}\x\clopen{0, \rho}$ is shaded. Since there are three points inside the rectangle, $e = 3$. Lines through the highest point in the rectangle and lowest point in $\timeseg{0}{1}\x\clopen{0, \rho}$ indicate $l$ and $u$ respectively.}\label{figure:tau_sampling}
\end{figure}
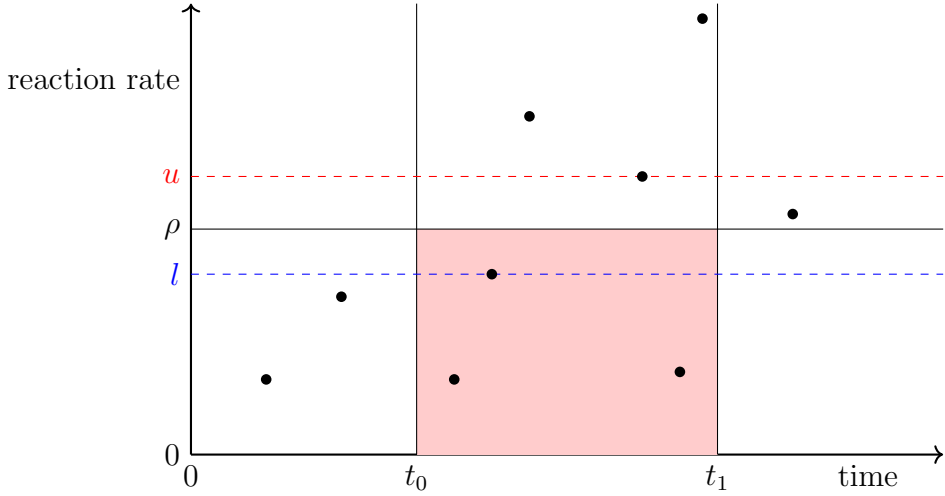

As required by Property~\ref{property:tau:distribution}, 
the value $e$ is Poisson distributed with rate $\rho t$. 
The value $l$ is the maximal among $e$ uniformly distributed points, 
and is thus distributed as $\rho\cdot\BetaVar{e}{1}$. 
The value $u$ is the first Point sampled by the Poisson process above $\rho$,
and is thus distributed as $\rho + \Exp{t}$.

\subsubsection{Implementation of Procedure~\ref{procedure:tau:split}}

Given a $\tau$-leaping step $\tau_{\timeseg{0}{1}} = \elh{_i}_i$ 
spanning $\timeseg{0}{1}$, 
we will sample a $\tau$-leaping step $\tau_{\timeseg{0}{1/2}} = {\elh{_i^L}}_i$
and a $\tau$-leaping precursor $\tau_{\timeseg{1/2}{1}} = {\elh{_i^R}}_i$
with initial state $S_{t_0}$.

For each reaction $R_i$, denote $\rho_i = \rho_i(S_{t_0})$.
Let $\cX$ be a Poisson point process on $\RR^2$ with rate 1,
and sample $X$ according to the distribution 
$\cX | f(\cX, \timeseg{0}{1}\x\clopen{0, \rho_i}) = \elh{_i}$.
We set ${\elh{_i^L}}$ to be $f(X, \timeseg{0}{1/2}\x\clopen{0, \rho_i})$ 
and ${\elh{_i^R}}$  to be $f(X, \timeseg{1/2}{1}\x\clopen{0, \rho_i})$ 
(Figure~\ref{figure:tau_split}).

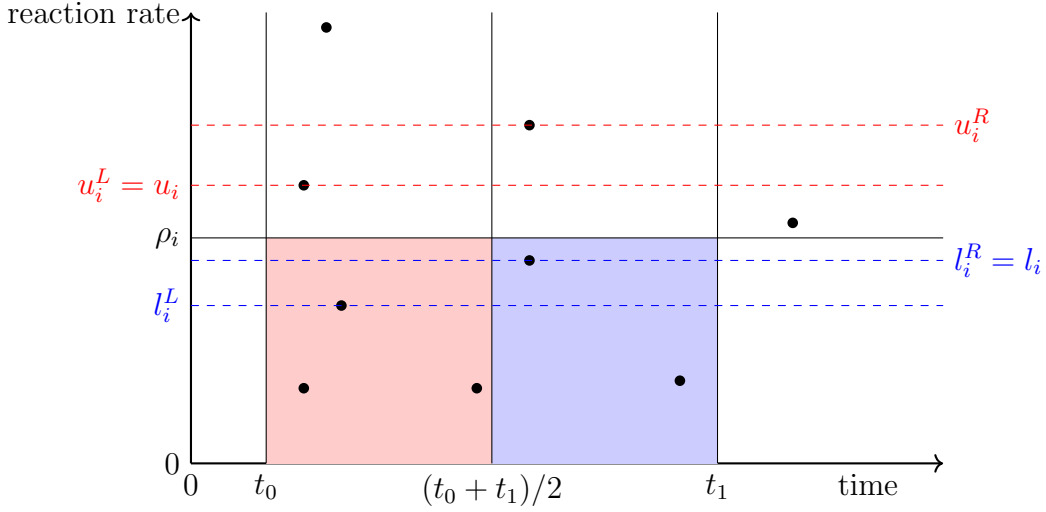
\begin{figure}[H]
\begin{tikzpicture}
\draw[->, thick] (0,0) -- (10,0);
\draw[->, thick] (0,0) -- (0,6);

\def\tzero{1};
\def\thalf{4};
\def\tone{7};
\def\rate{3};
\fill[red!20] (\tzero,0) rectangle (\thalf,\rate);
\fill[blue!20] (\thalf,0) rectangle (\tone,\rate);


\node[below] at (0, 0) {0};
\node[below] at (\tzero, 0) {$t_0$};
\node[below] at (\thalf, 0) {$(t_0 + t_1) / 2$};
\node[below] at (\tone, 0) {$t_1$};
\node[below] at (9,0) {time};

\node[left] at (0,0) {0};
\node[left] at (0,3) {$\rho_i$};
\node[left] at (0,6) {reaction rate};

\draw[black] (\tzero, 0) -- (\tzero, 6);
\draw[black] (\thalf, 0) -- (\thalf, 6);
\draw[black] (\tone, 0) -- (\tone, 6);
\draw[black] (0, \rate) -- (10, \rate);

\fill (1.5,1) circle (2pt);
\fill (1.8,5.8) circle (2pt);
\fill (3.8,1) circle (2pt); 
\fill (6.5,1.1) circle (2pt); 
\fill (8,3.2) circle (2pt);
\fill (2,2.1) circle (2pt);
\draw[blue, dashed] (0,2.1) -- (10,2.1);
\node[blue, left] at (0,2.1) {$l_i^L$};
\fill (4.5,2.7) circle (2pt); 
\draw[blue, dashed] (0,2.7) -- (10,2.7);
\node[blue, right] at (10,2.7) {$l_i^R = l_i$};
\fill (1.5,3.7) circle (2pt);
\draw[red, dashed] (0,3.7) -- (10,3.7);
\node[red, left] at (0,3.7) {$u_i^L = u_i$};
\fill (4.5,4.5) circle (2pt);
\draw[red, dashed] (0,4.5) -- (10,4.5);
\node[red, right] at (10,4.5) {$u_i^R$};

\end{tikzpicture}

\caption{An illustration of sampling ${\elh{_i^L}}$ and ${\elh{_i^R}}$ conditioned on $\elh{_i}$.}\label{figure:tau_split}
\end{figure}

The value $e_i^L$ is distributed $\Bin{e}{1/2}$, and $e_i^R = e_i - e_i^L$.
One of $u_i^L$ and $u_i^R$ is equal to $u_i$ with equal probability, and the other is distributed $u_i + \Exp{t/2}$.
One of $l_i^L$ and $l_i^R$ is equal to $l_i$ with probabilities $\frac{e_i^L}{e_i}$ and $\frac{e_i^R}{e_i}$ respectively.
Let $l_i^X$ be the one that is not equal to $l_i$.
Then $l_i^X$ is distributed $l_i \cdot \BetaVar{e_i^X}{1}$.

To prove that this procedure maintains Properties~\ref{property:tau:distribution} 
and~\ref{property:tau:independence},
note that when $\elh{_i}\sim\cY(\rho, (t_1 - t_0)) = f(\cX, \timeseg{0}{1}\x\clopen{0, \rho})$,
$X$ is a random sample of $\cX$, 
and is independent from any random variable from which $\elh{_i}$ are independent. 
Thus, by the properties of the Poisson process, ${\elh{_i^L}}$ and ${\elh{_i^R}}$
are random samples of $\cY(\rho, (t_1 - t_0)/2)$,
and are independent of each other and of any other random 
variable from which $\elh{_i}$ are independent.

\subsubsection{Implementation of Procedure~\ref{procedure:tau:resample}}

Given a $\tau$-leaping precursor $\tau_{\timeseg{0}{1}}^* = {\elh{_i^*}}_i$ 
spanning $\timeseg{0}{1}$ with initial state $S^*$, 
we will sample a $\tau$-leaping step $\tau_{\timeseg{0}{1}} = \elh{_i}_i$ 
spanning $\timeseg{0}{1}$ with initial state $S_{t_0}$.

For each reaction $R_i$, denote $\rho_i = \rho_i(S_{t_0})$ and $\rho_i^* = \rho_i(S^*)$.
Let $\cX$ be a Poisson point process on $\RR^2$ with rate 1,
and sample $X$ according to the distribution $\cX | f(\cX, \timeseg{0}{1}\x\clopen{0, \rho_i^*}) = {\elh{_i^*}}$.
We set $\elh{_i}$ to be $f(X, \timeseg{0}{1}\x\clopen{0, \rho_i})$.
(Figure~\ref{figure:tau_resampling}).
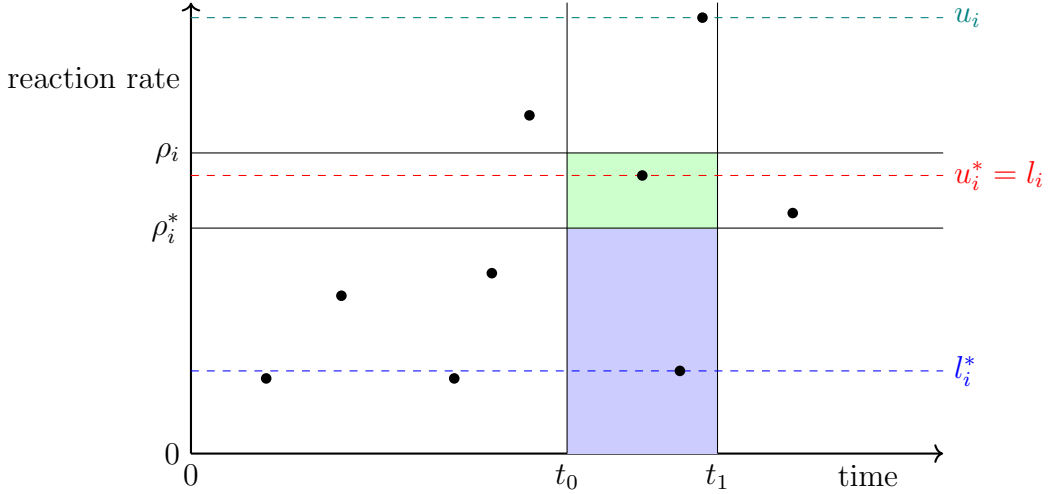
\begin{figure}[H]
\begin{tikzpicture}
\draw[->, thick] (0,0) -- (10,0);
\draw[->, thick] (0,0) -- (0,6);

\def\tzero{3};
\def\tzero{5};
\def\tone{7};
\def\rate{3};
\fill[blue!20] (\tzero,0) rectangle (\tone,\rate);
\fill[green!20] (\tzero,\rate) rectangle (\tone,4);


\node[below] at (0, 0) {0};
\node[below] at (\tzero, 0) {$t_0$};
\node[below] at (\tone, 0) {$t_1$};
\node[below] at (9,0) {time};

\node[left] at (0,0) {0};
\node[left] at (0,3) {$\rho_i^*$};
\node[left] at (0,4) {$\rho_i$};
\node[left] at (0,5) {reaction rate};

\draw[black] (\tzero, 0) -- (\tzero, 6);
\draw[black] (\tone, 0) -- (\tone, 6);
\draw[black] (0, \rate) -- (10, \rate);
\draw[black] (0, 4) -- (10, 4);
\draw[blue, dashed] (0,1.1) -- (10,1.1);
\node[blue, right] at (10,1.1) {$l_i^*$};
\draw[red, dashed] (0,3.7) -- (10,3.7);
\node[red, right] at (10,3.7) {$u_i^* = l_i$};
\draw[teal, dashed] (0,5.8) -- (10,5.8);
\node[teal, right] at (10,5.8) {$u_i$};

\fill (1,1) circle (2pt);
\fill (2,2.1) circle (2pt);
\fill (3.5,1) circle (2pt); 
\fill (4,2.4) circle (2pt); 
\fill (6.5,1.1) circle (2pt); 
\fill (6,3.7) circle (2pt); 
\fill (4.5,4.5) circle (2pt);
\fill (8,3.2) circle (2pt);
\fill (6.8,5.8) circle (2pt);

\end{tikzpicture}

\caption{An illustration of sampling $\elh{_i}$ 
conditioned on ${\elh{_i^*}}$.
Using the rate $\rho_i^*$ leads to the blue rectangle, and hence $e_i^* = 1$.
Using the rate $\rho_i$ leads to the union of the blue and green rectangles,
and hence $e_i = 2$.
}\label{figure:tau_resampling}
\end{figure}

To determine the distribution of $\elh{_i}$, we distinguish between three cases:

\begin{itemize}
    \item $\rho_i \in \clopen{u_i^*, \infty}$: In this case (See Figure~\ref{figure:tau_resampling})\begin{itemize}
        \item $e_i \sim e_i^* + 1 + \Pois{(t_1 - t_0)(\rho_i - u_i^*)}$;
        \item $l_i \sim u_i^* + (\rho_i - u_i^*)\cdot \BetaVar{e_i - e_i^* - 1}{1}$;
        \item $u_i \sim \rho_i + \Exp{t_1 - t_0}$.
    \end{itemize}
    \item $\rho_i \in \clopen{l_i^*, u_i^*}$: In this case \begin{itemize}
        \item $e_i = e_i^*$;
        \item $l_i = l_i^*$;
        \item $u_i = u_i^*$.
    \end{itemize}
    \item $\rho_i \in \clopen{0, l_i^*}$: In this case \begin{itemize}
        \item $e_i \sim \Bin{e_i^* - 1}{\frac{\rho_i}{l_i^*}}$;
        \item $l_i \sim \rho_i \cdot \BetaVar{e_i}{1}$;
        \item $u_i \sim \rho_i + \pr{l_i^* - \rho_i} \BetaVar{1}{e_i^* - e_i - 1}$.
    \end{itemize}
\end{itemize}

To prove that this procedure maintains Properties~\ref{property:tau:distribution} and~\ref{property:tau:independence},
note that when ${\elh{_i^*}}\sim\cY(\rho_i^*, t) = f(\cX, \timeseg{0}{1}\x\clopen{0, \rho_i^*})$,
$X$ is a random sample of $\cX$.
Thus, by the properties of the Poisson process, 
the tuple $\elh{_i}$ is a random sample of $\cY(\rho, t)$ 
and is independent from any random variable from which ${\elh{_i^*}}$ are independent.

\subsection{The stopping condition}\label{section:basic:stop_condition}

We will now define the exact stopping condition used in Procedure~\ref{procedure:stop_condition},
as well as a series of approximate stopping conditions.
We will then prove that the distribution sampled using the approximate stopping conditions converges to both $\PP_T$ and to the distribution sampled using the exact stopping condition,
thereby concluding the proof.


Note that splitting a $\tau$-leaping step by Procedure~\ref{procedure:tau:split}
conserves the total number of reaction occurrences and widens the gap between $l_i$ and $u_i$, 
and that resampling a step from a precursor by Procedure~\ref{procedure:tau:resample} 
preserves ${\elh{_i^*}}$ if $\rho_i \in \clopen{l_i^*, u_i^*}$.
Thus, once $\rho_i$ can no longer leave $\clopen{l_i, u_i}$, 
the total number of reaction occurrences will not change by further splitting.
A technical point we need is that the change in rate after the last reaction occurrence
in a time segment cannot reduce the total number of reactions. 
Thus, when computing a lower bound on the change in rate, we do not include the last
occurrence of $R_i$.

Recall that $O_i$ is the stoichiometry vector of the $i$'th reaction.
Let $O_i^+$ be the vector such that $O_i^+[k] = \max\pr{O_i[k], 0}$,
and $O_i^-$ be the vector such that $O_i^-[k] = \min\pr{O_i[k], 0}$, for every $k$.
We extend the domain of the rate function to $\ZZ^s$ by setting $\rho_i(S) = 0$ 
if $S[k] < 0$ for any $k$.

We can now define the stopping condition:

\begin{definition}[The exact stopping condition]\label{def:stop_condition}
    Let $\delta_i = \min\pr{e_i, 1}$. 
    A $\tau$-leaping step over $\timeseg{0}{1}$ with vectors $\pr{e_i}_i, \pr{l_i}_i, \pr{u_i}_i$ is mapped to STOP if and only if for all $i$:
    \begin{enumerate}[label=C\arabic*., ref=C\arabic*]
        \item\label{stop_condition_1} 
        The rate $\rho_i(S_{t_0} + \sum_j e_j O_j^+) < u_i$.
        \item\label{stop_condition_2}
        The rate $\rho_i(S_{t_0} + \sum_j e_j O_j^- - \delta_i O_i^-) \ge l_i$.
    \end{enumerate}
\end{definition}

The proof that the exact stopping condition satisfies Property~\ref{property:stop:exact} 
makes use of a binary tree describing the $\tau$-leaping steps.
As described in Section~\ref{section:basic_algorithm}, 
the algorithm constructs a binary tree, 
where the root node corresponds to the $\tau$-leaping step over the time segment $\clopen{0, T}$, 
and if the stopping condition is not satisfied by a $\tau$-leaping step over $\timeseg{0}{1}$,
then the left child node corresponds to the $\tau$-leaping step over $\timeseg{0}{1/2}$
and the right child node corresponds to the $\tau$-leaping step over $\timeseg{1/2}{1}$.

For a node $N$, define $e_i^N$ to be the value $e_i$ of the $\tau$-leaping step associated with $N$,
and define $l_i^N$, $u_i^N$, and $\delta_i^N$ similarly.

For a state $S$ and a node $N$,
let $N(S) = S + \sum e_i^N O_i$ be the approximate state after applying the step, 
let $N^+(S) = S + \sum e_i^N O_i^+$,
and let $N^-(S) = S + \sum e_i^N O_i^-$.
Note that for all $k$, $N^-(S)[k] \le S[k] \le N^+(S)[k]$, 
and that since the rate is monotone in the amount of the $k$'th reactant,
$\rho_i(N^-(S)) \le \rho_i(S) \le \rho_i(N^+(S))$.

We will refer to nodes mapped to STOP as \emph{stable}.
We will now introduce an approximate stopping condition parametrized by $d \in \NN$.
Under the $d$-approximate stopping condition, the binary tree constructed by the algorithm 
is the complete binary tree of depth $d$.

\begin{definition}[The $d$-approximate stopping condition]\label{def:stop_condition_approx}
    Let $d \in \NN$.
    A $\tau$-leaping step is mapped to STOP if and only if it corresponds to a node
    of depth at least $d$.
\end{definition}

Let $\PP_T^d$ be the distribution sampled by the algorithm using the $d$-approximate stopping condition
rather than the exact stopping condition.
We will now show that $\PP_T^d$ converges to $\PP_T$ as $d \to \infty$.

\begin{lemma}\label{lemma:distribution_convergence}
    If $\rho_i$ is bounded in probability for all $i$,
    then as $d \to \infty$, $\PP_T^d$ converges in $L_1$ to $\PP_T$.
\end{lemma}
\begin{proof}
    The state is only changed by nodes at depth $d$, and we may ignore all other nodes.
    For those nodes, the length of the time segment is $2^{-d}T$.
    By Properties~\ref{property:tau:distribution} and~\ref{property:tau:independence},
    the steps at nodes at depth $d$ are independent when conditioned on the state at the beginning of the node; 
    for each node $N$, the values $\pr{e_i^N}_i$ are independent when conditioned on the state at the beginning of the node;
    and the distribution of $e_i^N$ is $\Pois{2^{-d}T \rho_i^N}$.
    
    Since the reaction rates are bounded in probability,
    for every $\epsilon > 0$ there exists a $C$ such that $P(\forall t,\ \ C > \sum_i \rho_i(S_t)) > 1 - \epsilon$.
    Conditional on that event, the total number of reaction occurrences is stochastically dominated by $\Pois{CT}$. 
    Thus, the total number of reaction occurrences is bounded in probability.
    After a reaction occurrence at time $t$ in a node $N$, the rate $\rho_i(N)$ used by the algorithm 
    differs from the rate $\rho_i(S_t)$ for a period of time shorter than the length of the 
    $\tau$-leaping step in which it occurs, 
    which is $2^{-d}T$, and the absolute difference is smaller than $C$.
    Thus, the integral of the $L_1$ difference of the transition rates 
    of the system simulated by the algorithm and the discrete state SDE is in $O(2^{-d})$ in probability.
    Since the integral of the $L_1$ difference of transition rates is greater than the integral of the $L_1$
    difference of the distributions, the claim is proved.
\end{proof}

\bigskip

We now prove that for any $d$, the exact stopping condition~\ref{def:stop_condition} 
satisfies Property~\ref{property:stop:exact} for the system sampled using the
$d$-approximate algorithm.
We will prove this by induction on the depth of the node in the tree.

\begin{lemma}\label{lemma:approximate_node_stability}
    Let $N$ be a node spanning $\timeseg{0}{1}$ sampled using the $d$-approximate algorithm,
    with initial state $S_{t_0}$ and child nodes $L$ and $R$.
    If $N$ satisfies the exact stopping condition~\ref{def:stop_condition}, 
    then $L$ and $R$ satisfy it as well, 
    and $S_{t_1} = S_{t_0} + \sum_i e_i^N O_i$.
\end{lemma}
\begin{proof}
    For a node of depth $d$, the lemma is trivially true since the algorithm 
    changes $S$ to $S + \sum e_i^N O_i$.
    We will now show that if the lemma applies to all nodes of depth greater than $d'$, 
    then it applies to nodes of depth $d'$ as well.

    Consider a node $N$ of depth $d'$ satisfying~\ref{stop_condition_1} 
    and~\ref{stop_condition_2}, 
    with child nodes $L$ and $R$.
    We first prove that $L$ satisfies~\ref{stop_condition_1} and~\ref{stop_condition_2}. 
    For every $k$ we have\footnote{
        Recall that for every node $N$, and in particular for the node $L$, 
        we defined $N(S) = S + \sum_i e_i^N O_i$ and $N^+(S) = S + \sum_i e_i^N O_i^+$.
    }
    $L^+(S)[k] \le N^+(S)[k]$ and $L^-(S)[k]  - \delta_i^L O_i^-[k] \ge N^-(S)[k]  - \delta_i^N O_i^-[k]$,
    and therefore $\rho_i(L^+(S)) \le \rho_i(N^+(S)) \le u_i^N \le u_i^L$
    and $\rho_i(L^-(S) - \delta_i^L O_i^-) \ge \rho_i(N^-(S) - \delta_i^N O_i^-) \ge l_i^N \ge l_i^L$,
    and hence $L$ satisfies~\ref{stop_condition_1} and~\ref{stop_condition_2}.
    By the induction hypothesis, Lemma~\ref{lemma:approximate_node_stability} applies to $L$, 
    and therefore the state after $L$ is equal to $L(S)$.

    We turn to showing that the lemma applies to $R$.
    \begin{align*}
        u_i^N \ge & r_i \prod_{k \in I_i} \binom {S[k] + \sum_j e_j^N O^+_j[k]}{C_i[k]} \ge  \\
                  & r_i \prod_{k \in I_i} \binom {S[k] + \sum_j e_j^L O^+_j[k]}{C_i[k]} \ge \\
                  & r_i \prod_{k \in I_i} \binom {S[k] + \sum_j e_j^L O_j[k]}{C_i[k]} = \\
                  & r_i \prod_{k \in I_i} \binom {L(S)[k]}{C_i[k]} = \rho_i^R,
    \end{align*}
    where the first inequality is satisfied by Condition~\ref{stop_condition_1},
    the second because $e_j^N \ge e_j^L$, the third because $O^+_j[k] \ge O_j[k]$ at every component, 
    and the last equality follows from the induction hypothesis.  
    Thus, $\rho_i^R \le u_i^N \le u_i^R$, and therefore $e_i^L + e_i^R \le e_i^N$.
    Hence:
    \begin{align*}
        u_i^N > & r_i \prod_{k \in I_i} \binom {S[k] + \sum_j e_j^N O^+_j[k]}{C_i[k]} \ge  \\
                & r_i \prod_{k \in I_i} \binom {S[k] + \sum_j (e_j^L + e_j^R) O^+_j[k]}{C_i[k]} \ge \\
                & r_i \prod_{k \in I_i} \binom {S[k] + \sum_j e_j^L O_j[k] + \sum_j e_j^R O_j^+[k]}{C_i[k]} = \\
                & \rho_i\pr{L(S) + \sum_j e_j^R O_j^+},
    \end{align*}

    and hence $R$ satisfies~\ref{stop_condition_1}.

    We next consider~\ref{stop_condition_2}.

    If $e_i^R \ne 0$, then $\delta_i^R = \delta_i^N$, and:
    \begin{align*}
        l_i^N \le & r_i \prod_{k \in I_i} \binom {S[k] + \sum_j e_j^N O^-_j[k] - \delta_i^N O_i[k]}{C_i[k]} \le \\
                  & r_i \prod_{k \in I_i} \binom {S[k] + \sum_j (e_j^L + e_j^R) O^-_j[k] - \delta_i^R O_i[k]}{C_i[k]} \le \\
                  & r_i \prod_{k \in I_i} \binom {S[k] + \sum_j e_j^L O_j[k] + \sum_j e_j^R O_j^-[k] - \delta_i^R O_i[k]}{C_i[k]} = \\
                  & \rho_i\pr{L(S) + \sum_j e_j^R O_j^- - \delta_i^R O_i[k]},
    \end{align*}
    showing that $R$ satisfies~\ref{stop_condition_2}.
    Otherwise, since $e_i^R = 0$, we have $l_i^R = 0$, so $\rho_i\pr{L(S) + \sum_j e_j^R O_j^- - \delta_i^R O_i[k]} \ge l_i^R$.
    
    Since $R$ satisfies~\ref{stop_condition_1} and~\ref{stop_condition_2}, 
    the induction hypothesis applies to $R$, and the state after the algorithm simulates $R$ is 
    $R(L(S)) = S + \sum_i (e_i^R + e_i^L) O_i = S + \sum_i (e_i^{R*} + e_i^L) O_i = S + \sum_i e_i^N O_i = N(S)$.
\end{proof}

\bigskip

Lemma~\ref{lemma:approximate_node_stability} shows 
that the distribution sampled using the inexact stopping condition is the same 
as the distribution sampled using the stopping condition that stops when either 
the inexact or the exact stopping condition map a step to STOP.\@
We will now show that as $d \to \infty$, almost surely this stopping condition is equal to the exact stopping condition.

\begin{lemma}\label{lemma:finite_unstable}
    Almost surely, there exist only a finite number of nodes that do 
    not satisfy the stopping condition~\ref{def:stop_condition}.
\end{lemma}
\begin{proof}
    As we saw in the proof of Lemma~\ref{lemma:distribution_convergence},  the total number of reaction occurrences is bounded in probability,
    and for every node $N$ of depth $d$, the values $e_i^N$ 
    are distributed independently as $\Pois{2^{-d}T \rho_i^N}$.
    Thus, the total number of occurrences in a node is distributed as 
    $\Pois{2^{-d}T \sum_i \rho_i^N}$.
    
    Since the reaction rates are bounded in probability, 
    for every $\epsilon$ there exists a $C$ such that the sum of the reaction rates is smaller than $C$
    with probability $1-\epsilon$.
    Conditioning on this event, 
    the total number of reaction occurrences in a node is stochastically dominated by $\Pois{2^{-d} C}$.
    The probability of a node having one occurrence is in $O(2^{-d})$, 
    and the probability of a node having at least two occurrences is in $O(2^{-2d})$.
    Since there are $2^d$ nodes, 
    and the probability of any node having at least two occurrences is in $O(2^{-d})$,
    almost surely there exists a $d$ where no node has more than one occurrence.

    When a node has zero occurrences it trivially satisfies~\ref{stop_condition_1} and~\ref{stop_condition_2}.
    When a node has only one occurrence, then $l_i^N \ne 0$ only for the reaction where $e_i^N = 1$.
    For this reaction, $e_i = \delta_i = 1$, and thus $S^- - \delta_i O_i^- = S$,
    so the node satisfies~\ref{stop_condition_2}.

    For a node $N$ with initial state $S$ and $e_i^N = 1$, 
    we need $\rho_i(N^+(S))$ to be smaller than $u_i$ for the node to satisfy~\ref{stop_condition_1}.
    Since $u_i - \rho_i$ is exponentially distributed with rate $O(2^{-d})$ in probability, 
    the probability of this condition not being satisfied is in $O(2^{-d})$.
    Since the number of nodes of a given depth with one reaction occurrence is bounded in probability
    and the probability of a node not satisfying Conditions~\ref{stop_condition_1} and~\ref{stop_condition_2} is $O(2^{-d})$,
    there almost surely exists a depth at which all nodes satisfy~\ref{stop_condition_1} and~\ref{stop_condition_2}
    and we are done.
\end{proof}

By Lemma~\ref{lemma:finite_unstable}, as $d \to \infty$, almost surely there exists a depth $d'$ at which the recursion (and the algorithm) terminates. 
For any $d > d'$, by Lemma~\ref{lemma:approximate_node_stability}, the $d$-inexact algorithm behaves the same as the exact algorithm.
Therefore, as $d \to \infty$, the $d$-inexact algorithm converges to the exact algorithm. 
Hence, Property~\ref{property:stop:exact} follows from Lemma~\ref{lemma:approximate_node_stability}, 
and Theorem~\ref{theorem:exact_algorithm_correctness} follows from Lemma~\ref{lemma:distribution_convergence}.

\section{The full \texorpdfstring{$\tau$}{tau}-splitting algorithm}\label{section:tau_improved}

We will now present the full $\tau$-splitting algorithm.
In the basic algorithm, when we split a node, we split all the reactions,
even if some reactions satisfy the conditions of the exact stopping condition~\ref{def:stop_condition} individually.
To improve the complexity of the algorithm, 
we identify a stopping condition for individual reactions,
and stop splitting each reaction as early as we can.

    



\subsection{The stopping condition for individual reactions}

We will now present the stopping condition for individual reactions 
used in the full $\tau$-splitting algorithm.

Let $N$ be a node spanning $\timeseg{0}{1}$. 
Since the progress of a discrete state SDE is unpredictable, 
it is difficult to find an exact stopping condition that uses only the data from before $t_0$
to decide whether to stop splitting an individual reaction.
Instead, we will use a stopping condition that also relies on a condition being maintained during the
time segment $\timeseg{0}{1}$.
We will stop splitting the reaction once it satisfies the condition, 
and resume splitting it if the condition is violated.

We will use the following condition:

\begin{definition}[The exact stopping condition for reactions]\label{def:stop_condition_reaction}
    A reaction $R_i$ is mapped to STOP if and only if it satisfies Conditions~\ref{stop_condition_1} and~\ref{stop_condition_2}.
\end{definition}

We now prove the analog of Property~\ref{property:stop:exact} for the stopping condition:

\begin{lemma}[Reaction Stopping Condition]\label{lemma:reaction_deactivation}
    Let $R_i$ be a reaction that satisfies~\ref{stop_condition_1} and~\ref{stop_condition_2} at a node $N$.
    If for all right child nodes\footnote{
        A \emph{right child node} is a node that is the right child of its parent node.
    } $M$ in the subtree rooted by $N$, all reactions $R_j$ on which $R_i$ depends satisfy
    $e_j^{M*} \le e_j^M$, then the total number of times the reaction $R_i$ occurs in the time segment spanned by $N$ is equal to $e_i^N$.
\end{lemma}
\begin{proof}
    Let $\timeseg{0}{1}$ be the time segment spanned by $N$, 
    and let $E_{j, t}$ be the number of events of the $j$'th reaction at the time segment $\clopen{t_0, t}$.

    Since all the reactions $R_j$ on which $R_i$ depends satisfy $e_j^{M*} \le e_j^M$ 
    at all nodes in the subtree,
    we have $E_{j, t} \le e_j^N$ for all $t \in \timeseg{0}{1}$.
    Thus, for a node $M$ spanning the time-segment $\timeseg{a}{b}$:
    \begin{align*}
        u_i^M > & r_i \prod_{k \in I_i} \binom {S[k] + \sum_j e_j^N O^+_j[k]}{C_i[k]} \ge  \\
                  & r_i \prod_{k \in I_i} \binom {S[k] + \sum_j E_{j, t_b} O^+_j[k]}{C_i[k]} \ge \\
                  & r_i \prod_{k \in I_i} \binom {S[k] + \sum_j E_{j, t_a} O_j[k] + e_j^M O_j^+[k]}{C_i[k]} = 
                  \rho_i\pr{S_{t_a} + \sum e_j^M O_j^+},
    \end{align*}
    where the first inequality is due to~\ref{stop_condition_1}, 
    the second follows by the Lemma's hypothesis, 
    and the third follows since $O^+_j[k] \ge O_j[k]$.
    Thus, all nodes in the subtree rooted by $N$ satisfy~\ref{stop_condition_1} for $R_i$.
    If $e_i^M > 0$, then $\delta_i^M = \delta_i^N$, and similarly: 
    \begin{align*}
        l_i^M \le & r_i \prod_{k \in I_i} \binom {S[k] + \sum_j e_j^N O^-_j[k] - \delta_i^N O_i^-}{C_i[k]} \le  \\
        & r_i \prod_{k \in I_i} \binom {S[k] + \sum_j E_{j, t_b} O^-_j[k] - \delta_i^M O_i^-}{C_i[k]} \le \\
        & r_i \prod_{k \in I_i} \binom {S[k] + \sum_j E_{j, t_a} O_j[k] + e_j^M O_j^-[k] - \delta_i^M O_i^-}{C_i[k]} = 
        \rho_i\pr{S_{t_a} + \sum e_j^M O_j^- - \delta_i^M O_i^-}.
    \end{align*}
    Otherwise, $l_i^M = 0$, and in particular $l_i^M \le \rho_i\pr{S_{t_a} + \sum e_j O_j^- - \delta_i^M O_i^-}$.
    Therefore, all nodes in the subtree rooted by $N$ satisfy~\ref{stop_condition_2} for $R_i$.

    Consider the subtree rooted by the node $N$ of the tree sampled by the 
    $d$-approximate algorithm defined in Section~\ref{section:basic:stop_condition}. 
    As in the proof of Lemma~\ref{lemma:approximate_node_stability}, 
    when all nodes in this subtree satisfy Conditions~\ref{stop_condition_1} 
    and~\ref{stop_condition_2} for the reaction $R_i$,
    the total number of occurrences of $R_i$ in the time segment spanned by $N$ 
    is equal to $e_i^N$.
    As in the last paragraph of Section~\ref{section:basic:stop_condition}, as $d \to \infty$, 
    the inexact stopping condition converges to the exact stopping condition for individual reactions, 
    proving the claim.
\end{proof}

\subsection{The full \texorpdfstring{$\tau$}{tau}-splitting algorithm}

We will now describe an algorithm that makes use of Lemma~\ref{lemma:reaction_deactivation} 
to stop splitting individual reactions.
We will refer to reactions that we have stopped sampling as \emph{inactive}, 
to the marking of a reaction as inactive as \emph{deactivation},
and to the marking of an inactive reaction as active as \emph{reactivation}.
Reactions will be deactivated while the algorithm processes a node, 
and will remain associated with that node.
We will refer to a reaction that satisfies~\ref{stop_condition_1} and~\ref{stop_condition_2} as \emph{stable}.

Lemma~\ref{lemma:reaction_deactivation} allows us to stop splitting a stable reaction $R_i$,
so long as we resume splitting it reteroactively if one of the reactions $R_j$ 
it depends on has $e_j > e_j^*$,
which may happen if $R_j$ is unstable.
Instead of storing $\tau$-leaping steps and precursors as vectors, 
we store them as sparse vectors, containing tuples $(i, e_i, l_i, u_i)$.
We store the data for the inactive reactions in a way that allows us to 
reactivate them when changes in the system invalidate the condition of 
Lemma~\ref{lemma:reaction_deactivation}.

Let the \emph{current node} be the node the algorithm is currently operating on.
Let $P$ be the path from the root to the current node.

We maintain the following data structures:
\begin{enumerate}[label=D\arabic*., ref=D\arabic*]
    \item\label{ds:dependent_reactions} \textbf{The Dependent Reaction Sets:} 
    For every reactant, the set of unstable reactions having the reactant as an input.
    \item\label{ds:reaction_list} \textbf{The Node Reaction Lists:} 
    For every node $N$, a list $\rdatalist{N}$ storing:
    \begin{itemize}
        \item The active reactions or reaction precursors in $N$.
        \item The reactions deactivated in $N$.
    \end{itemize}
    \item\label{ds:inactive_index} \textbf{The Inactive Index:}
    A list storing for every reaction either a null value if it is active, 
    or the node in which the reaction was deactivated.
    \item\label{ds:is_stable}\textbf{The Reaction Stability List:}
    A list storing for every reaction whether it is currently stable.
    \item\label{ds:node_status}\textbf{The Node Status List:}
    A list storing for every node whether the recursive algorithm is done operating on it,
    whether it is in $P$, or whether the algorithm has not yet reached it.
    \item\label{ds:bounding_states}\textbf{The Bounding States}
    $S^+ = S + \sum_{N \in P, i \in \rdatalist{N}} e_i^N O_i^+$ 
    and $S^- = S + \sum_{N \in P, i \in \rdatalist{N}} e_i^N O_i^-$.
\end{enumerate}

We perform a recursive algorithm starting with the root node.
We describe the operation of the algorithm at a node $N$ with left and right child nodes $L$ and $R$, respectively.

\paragraph{Step 1.}\label{algorithm:lazy:update_node_list}
We start by updating the Node Status List (\ref{ds:node_status}) to mark that $N$ is in $P$.

\paragraph{Step 2.}\label{algorithm:lazy:backward_reactivation}
If $N$ is a right child node, the Node Reaction List (\ref{ds:reaction_list}) 
stores the $\tau$-leaping precursor, which has the tuples $(i, e_i^{N*}, l_i^{N*}, u_i^{N*})$.
We sample the values $\elh{_i^N}$ using Procedure~\ref{procedure:tau:resample} based on $\rho_i(S)$
and $(e_i^{N*}, l_i^{N*}, u_i^{N*})$. 
For every reaction $R_i$ where $e_i^{N*} > e_i^N$ we go over all reactions 
depending on $R_i$ and reactivate them.

\paragraph{Step 3.}\label{algorithm:lazy:update_bounds}
For all reactions $R_i$ in $F^N$, we add $e_i O_i^+$ to $S^+$ and $e_i O_i^-$ to $S^-$.

\paragraph{Step 4.}\label{algorithm:lazy:forward_reactivation}
We check for every reaction $R_i$ if it is stable. 
We update the stability of all reactions in the 
Dependent Reaction Sets (\ref{ds:dependent_reactions}) 
and the Reaction Stability List~(\ref{ds:is_stable}), 
and reactivate all inactive reactions on which
an unstable reaction depends.

\paragraph{Step 5.}\label{algorithm:lazy:deactivation}
We check which reactions can be deactivated using the 
Dependent Reaction Sets~(\ref{ds:dependent_reactions}), 
by checking for every output of the reaction 
if the set of unstable reactions depending on it is empty.
We store the list of reactions that can be deactivated 
in the Node Reaction List~(\ref{ds:reaction_list}) of the node $N$,
and mark their location in the Inactive Index~(\ref{ds:inactive_index}).

\paragraph{Step 6.}\label{algorithm:lazy:apply_or_recurse}
If all reactions have been deactivated, 
we update the state from $S$ to $S + \sum_{i \in \rdatalist{N}} e_i^N O_i$,
$S^+$ to $S^+ + \sum_{i \in \rdatalist{N}} e_i^N (O_i - O_i^+)$,
and $S^-$ to $S^- + \sum_{i \in \rdatalist{N}} e_i^N (O_i - O_i^-)$.
We then finish the recursive call.

Otherwise, for all active reactions $R_i$, 
we subtract $e_i O_i^+$ from $S^+$ and $e_i O_i^-$ from $S^-$,
split them in two using Procedure~\ref{procedure:tau:split}, and add the two halves to the Node Reaction Lists~(\ref{ds:reaction_list}) $\rdatalist{L}$ and $\rdatalist{R}$.

\paragraph{Step 7.}\label{algorithm:lazy:recurse_left}
We recursively apply the algorithm to $\tau_{\timeseg{0}{1/2}}$. 

\paragraph{Step 8.}\label{algorithm:lazy:recurse_right}
We recursively apply the algorithm to $\tau_{\timeseg{1/2}{1}}$. 

\paragraph{Step 9.}\label{algorithm:lazy:recursion_end}
For each reaction $R_i$ still in $\rdatalist{N}$,  
we update the state from $S$ to $S + e_i^N O_i$,
$S^+$ to $S^+ + e_i^N (O_i - O_i^+)$,
and $S^-$ to $S^- + e_i^N (O_i - O_i^-)$.

\bigskip

This concludes the recursive procedure. We now explain the reactivation of reactions.

\paragraph{Reactivating reactions.}\label{algorithm:lazy:reactivation}

Reactivation of a reaction is also done recursively.
To reactivate a reaction, we use the Inactive Index~(\ref{ds:inactive_index}) 
to identify the node $N \in P$ in which it was deactivated,
and remove the reaction data $\elh{_i^N}$ 
from $\rdatalist{N}$~(\ref{ds:reaction_list}), from the Inactive Index~(\ref{ds:inactive_index}), 
and from the Bounding States~(\ref{ds:bounding_states}).
In the recursive step of the reactivation,
if $N$ is the current node, 
we add $\elh{_i^N}$ to $\rdatalist{N}$~(\ref{ds:reaction_list}) 
and to the Bounding States~(\ref{ds:bounding_states}) and complete the reactivation.
Otherwise, we split $\elh{_i^N}$ into two halves 
${\elh{_i^L}}$ and $\elh{_i^{R*}}$ 
using Procedure~\ref{procedure:tau:split}.
Since $P$ is a path from the root to the current node, it has one node in every depth. 
Thus, one of the child nodes is in $P$ and the other is not.
If $L$ is not in $P$, then the algorithm has already finished processing it.
We a $e_i^L O_i$ to $S$, $S^+$, and $S^-$.
Since the reaction was inactive at the initial timepoint of $R$,
we have $\elh{_i^{R*}} = {\elh{_i^R}}$,
and recursively reactivate ${\elh{_i^R}}$.
Otherwise, $R$ is not in $P$, and we add $\elh{_i^{R*}}$ to $\rdatalist{R}$~(\ref{ds:reaction_list})
and recursively reactivate ${\elh{_i^L}}$.

This concludes the reactivation procedure.

\subsection{The complexity of the full \texorpdfstring{$\tau$}{tau}-splitting algorithm}\label{section:full:complexity}

We will now compute algorithm's runtime densities of the RSSA-CR, BlSSSA, 
and the $\tau$-splitting algorithms.

We will begin with computing the runtime densities of the 
RSSA-CR algorithm~\cite{thanh2016rssacr} 
and the BlSSSA algorithm\footnote{Recall that $s$ is the dimension of the state space.}~\cite{ghosh2021blsssa}.
\begin{lemma}
    Let $D_i'$ be the number of reactions depending on $R_i$.
    The runtime density of the RSSA-CR algorithm
    is $\bigo{\sum_i (D_i' + K)\rho_i(S)}$,
    where $K$ is the number of groups used by the algorithm,
    and the runtime density of the BlSSSA algorithm
    is $\bigo{\sum_i \sqrt{s}\rho_i(S)}$.
\end{lemma}
\begin{proof}
    By the complexity analysis of~\cite{thanh2016rssacr}, 
    the runtime of the RSSA-CR algorithm when sampling the reaction $R_i$ is $O(D_i' + K)$,
    and by the complexity analysis of~\cite{ghosh2021blsssa},
    the runtime of the BlSSSA when sampling any reaction is $\sqrt s$.
    In a time segment of length $dt$ with initial state $S$, 
    the expected number of times the reaction $R_i$ is sampled
    is $\rho_i(S_t)dt + O(dt^2)$.
    Thus, the expected runtime of RSSA-CR when sampling a trajectory is the sum of the contribution of all reactions
    when integrating over all states in all trajectories:
    \begin{align*}
        \sum_i \int_{S} (D_i' + K)\rho_i(S)dS = \int_{S} \sum_i (D_i' + K)\rho_i(S)dS,
    \end{align*}
    and the expected runtime of BlSSSA when sampling a trajectory is
    \begin{align*}
        \sum_i \int_{S} \sqrt{s}\rho_i(S)dS = \int_S \sum_i \sqrt{s}\rho_i(S),
    \end{align*}
    proving the claim.
\end{proof}

The $\tau$-splitting algorithm relies on the approximation of the change of state with a
set of independent Poisson processes to allow multiple reaction occurrences in a single $\tau$-leaping step.
The higher the values of $S[k]$ are, the more accurate the approximation is, 
and more reactions can occur in a single $\tau$-leaping step.
When the reaction rates change too quickly, the first order approximation is highly inaccurate.
In this case the algorithm can advance reactions only one occurrence at a time,
making the complexity similar to the Gillespie algorithm.

We will define a function that will be a valid runtime density for the $\tau$-splitting algorithm.
The following result characterizes this function.
The contribution of an individual reaction to the function has two bounds, 
depending on the magnitude of $S[k]$.
We state the runtime density when all $S[k]$ are large, 
and then state how small $S[k]$ affect the individual reactions.

\begin{theorem}\label{theorem:full:complexity}
    Let $D_i$ be the total number of reactions depending on $R_i$ or on which $R_i$ depends.
    Let $\rho_k^\pm = \sum_j \rho_j \abs{O_j[k]}$.
    Let $M_k = \pr{\max_i \abs{O_i[k]}}_k$.

    When for all $k$ and $i$, $S[k] > \max\pr{e^2 M_k^2 \frac{\sum_{k \in I_i}\rho_k^\pm}{\rho_i}, e^2 M_k^2, 2C_i[k]}$, 
    the function
    \begin{align*}
        f(S) = \bigo{\sum_i D_i \sqrt{\rho_i \sum_{k \in I_i} \frac{\rho_k^\pm}{S[k]}}}
    \end{align*}
    is a runtime density for the $\tau$-splitting algorithm.

    More generally, when $S[k] > \max\pr{e^2 M_k^2 \frac{\sum_{k \in I_i}\rho_k^\pm}{\rho_i}, e^2 M_k^2, 2C_i[k]}$ for all $k \in I_i$,
    the function
    \begin{align}
        f_i(S) = \bigo{D_i \sqrt{\rho_i \sum_{k \in I_i} \frac{\rho_k^\pm}{S[k]}}}\label{eq:full:good_complexity}
    \end{align}
    is the contribution of $R_i$ to the runtime density.
    Otherwise, let $\cG_i$ be the set of reactions on which $R_i$ depends.
    The function
    \begin{align}
        f_i(S) = \bigo{D_i \sum_{j \in \cG_i}\rho_j}\label{eq:full:bad_complexity}
    \end{align}
    is the contribution of $R_i$ to the runtime density.
\end{theorem}

We prove the theorem by estimating the contribution of each reaction $R_i$ to the runtime of the algorithm.
The reactions are split by the $\tau$-splitting algorithm so long as they are not stable according to 
Condition~\ref{def:stop_condition_reaction}.
After that, they may still be split further because there are other unstable reactions depending on them,
but we count those splits as part of the contribution of those other reactions to the runtime.
A reaction that seems stable at the beginning of the time segment $\timeseg{a}{b}$ 
might also be split because it became unstable during $\timeseg{a}{b}$ due to the effect of a reaction on which it depends.

Let $\overline U_i(S, t)$ be the indicator random variable that is 1
if and only if $R_i$ is split due to its instability 
in a time segment of length $t$ with an initial state $S$,
and let $U_i(S, t)$ be the indicator random variable that is 1 
if and only if $R_i$ is unstable in a recursive step with initial state $S$ spanning a time segment of length $t$.
The difference between the two is that $\overline U_i(S, t)$ is 1 if the reaction is stable
according to Condition~\ref{def:stop_condition_reaction}, 
but is split due to it becoming unstable at a greater depth,
while $U_i(S, t)$ is 0 in that case.

We use the following inequality to relate $\overline U_i(S, t)$ and $U_i(S, t)$:
\begin{lemma}
    \label{lemma:unstable_inactive_inequality}
    Let $\timeseg{0}{1}$ be a time segment of length $T$, and let $S_t$ be the state at time $t$
    for some trajectory. Then:
    \begin{align*}
        \overline U_i(S_{t_0}, T) \le \sum_{d = 0}^\infty \sum_{n = 0}^{2^d} U_i(S_{t_0 + n2^{-d}T}, 2^{-d}T)
    \end{align*}
\end{lemma}
\begin{proof}
    The indicator $\overline U_i(S_{t_0}, T)$ is 1 when the reaction $R_i$ 
    is split during the time segment $\timeseg{0}{1}$ due to its instability.
    For that to happen, the $R_i$ has to be unstable at some dyadic split of $\timeseg{0}{1}$.
    The Lemma then follows from the union bound over these events.
\end{proof}

\begin{corollary}
    \label{corollary:unstable_inactive_integral_inequality}
    Let $\clopen{0, T}$ be a time segment with initial state $S_0$. 
    Then, the expectation of $\overline U_i(S, T)$ satisfies:
    \begin{align*}
        \EE\pr{\overline U_i(S_0, T)} \le \int_S \sum_{d = 0}^\infty 2^d U_i(S, 2^{-d}T)dS
    \end{align*}
\end{corollary}
\begin{proof}
    We apply Lemma~\ref{lemma:unstable_inactive_inequality} and integrate over all states $S$ 
    in all trajectories.
    Since we integrate over all trajectories, each $U_i(S_t, 2^{-d}T)$ 
    contributes at most a $\int_S U_i(S, 2^{-d}T)dS$ term.
    Summing all these terms proves the Corollary.
\end{proof}

We will now find bounds for $U_i(S, t)$.

\subsubsection{A bound for \texorpdfstring{$U_i(S, t)$}{Ui(S, t)} when \texorpdfstring{$S[k]$}{S[k]} are large}

This section is devoted to the proof of the following bound on $U_i(S, t)$:
\begin{lemma}\label{lemma:full:good_complexity}
    Let $\cR$ be the set of reactions, let $R_i = (I_i, C_i, O_i, r_i)$ be a reaction with input set $I_i$, input count vector $C_i$, 
    stoichiometry vector $O_i$, and rate constant $r_i$,
    and let $\rho_i$ be the rate of $R_i$.
    There exists some constant $c$ depending only on $\cR$,
    such that when\footnote{Here, $e$ is the base of the natural logarithm.} 
    $S[k] > \max\pr{e^2 M_k^2 \frac{\sum_{k \in I_i}\rho_k^\pm}{\rho_i}, e^2 M_k^2, 2C_i[k]}$ for all $k \in I_i$, we have:
    \begin{align}
        \EE\pr{U_i\pr{S, t}} \le \rho_i t^2 c \sum_{k \in I_i} \frac{\rho_k^\pm}{S[k]}. \label{eq:full:uist_bound}
    \end{align}
\end{lemma}

The random variable $U_i\pr{S, t}$ is an indicator random variable, and thus $\EE\pr{U_i\pr{S, t}} \le 1$.
Thus, Lemma~\ref{lemma:full:good_complexity} holds for any $t \ge \sqrt{\rho_i t^2 c \sum_{k \in I_i} \frac{\rho_k^\pm}{S[k]}}^{-1}$,
and we only have to prove~\eqref{eq:full:uist_bound} for $t < \sqrt{\rho_i c \sum_{k \in I_i} \frac{\rho_k^\pm}{S[k]}}^{-1}$.

\bigskip

Let $\Delta_k = \sum_i e_i \abs{O_i[k]}$ be a bound on the change of $S[k]$ in a $\tau$-leaping step.
We will next prove that the probability that $\Delta_k \ge S[k]$ is negligible:
\begin{lemma}\label{lemma:full:sk_large}
    Suppose that $t < \sqrt{\rho_i c \sum_{k \in I_i} \frac{\rho_k^\pm}{S[k]}}^{-1}$.
    Under the hypothesis of Lemma~\ref{lemma:full:good_complexity},
    $P\pr{\Delta_k \ge S[k]} \in o\pr{\rho_i t^2 \sum_{k \in I_i} \frac{\rho_k^\pm}{S[k]}}$.
\end{lemma}

\begin{proof}
    We will first show that the probability of $\Delta_k \ge S[k]$ 
    is dominated by the probability that $\Delta_k = S[k]$.

    For a Poisson random variable with rate $\lambda$ we have $P(\Pois{\lambda} = k) = \frac{\lambda^k e^{-\lambda}}{k!}$.
    When $k \ge \alpha \lambda$ for some constant $\alpha > 1$, increasing $k$ by 1 decreases $P(\Pois{\lambda} = k)$ 
    by at least $\frac{1}{\alpha}$.
    Thus, 
    \begin{align*}
        P(\Pois{\lambda} \ge \alpha \lambda) 
        = \sum_{k = \alpha \lambda}^\infty P(\Pois{\lambda} = k) 
        \le P(\Pois{\lambda} = \alpha\lambda) \sum_{k=0}^\infty \alpha^{-k} 
        = \frac{1}{1 - \alpha^{-1}} P(\Pois{\lambda}).
    \end{align*}

    In our case, the random variable $\Delta_k$ is stochastically dominated by $M_k\Pois{\rho_k^\pm t}$.
    Thus, $P(M_k\Delta_k \ge S[k]) \le P(\Pois{\rho_k^\pm t} \ge S[k]/M_k)$.
    Since $t \le \sqrt{\rho_i \sum_{k \in I_i}\frac{\rho_k^\pm}{S[k]}}^{-1}$,
    we have 
    \begin{align*} 
        \rho_k^\pm t 
         \le \rho_k^\pm \sqrt{\rho_i \sum_{k \in I_i}\frac{\rho_k^\pm}{S[k]}}^{-1}
         \le \rho_k^\pm \sqrt{\rho_i \frac{\rho_k^\pm}{S[k]}}^{-1} 
        = \sqrt{\frac{\rho_k^\pm S[k]}{\rho_i}}.
    \end{align*}
    Since $S[k] \ge e^2 M_k^2 \frac{\rho_k^\pm}{\rho_i}$ we have 
    $S[k] / M_k \ge e \sqrt{\frac{\rho_k^\pm S[k]}{\rho_i}} \ge e \rho_k^\pm t$, which is $e$ times the rate of the Poisson random variable.
    Thus, the probability that $\Delta_k \ge S[k]$ is dominated by the probability that $\Pois{\rho_k^\pm t} = \floor{S[k]/M_k}$.
    Let $A_k$ be $\floor{S[k]/M_k}$.
    
    \begin{align*}
        P\pr{\Pois{\rho_k^\pm t} = A_k}
        = \frac{\pr{\rho_k^\pm t}^{A_k} e^{-\pr{\rho_k^\pm t}}}{A_k!} 
        \le \frac{\pr{\rho_k^\pm t}^{A_k}}{A_k!} 
        \le \frac{\pr{e \rho_k^\pm t}^{A_k}}{A_k^{A_k}},
    \end{align*}
    Where the first inequality follows since the exponent of a negative value is smaller than 1,
    and the second by Stirling's approximation.
    To prove the lemma, we need to prove that $\frac{\pr{e \rho_k^\pm t}^{A_k}}{A_k^{A_k}} \in o(\rho_i t^2 \sum_{k \in I_i} \frac{\rho_k^{\pm}}{S[k]})$.
    We have already shown that $\frac{\pr{e \rho_k^\pm t}}{A_k} \le 1$, 
    so the inequality holds for the maximal value of $t$, 
    which is $\sqrt{\rho_i \sum_{k \in I_i} \frac{\rho_k^\pm}{S[k]}}^{-1}$, 
    since at that $t$ the right hand side is in $\bigo 1$.
    As $t$ decreases, the left hand side decreases as $t^{-A_k} \le t^{-3}$, 
    so the inequality holds in general.
\end{proof}

Lemma~\ref{lemma:full:sk_large} shows that for all $k$, $\Delta_k \le S[k]$ with a high probability.
This allows us to use the expression $\rho_i(S^-) = r_i \prod_k \binom{S[k] - \Delta_k}{C_i[k]}$ 
for the rate without having to treat the case where some of terms in the product become negative,
and will be used in expectation calculations later on.

\bigskip

For $R_i$ to stabilize we need $\rho_i(S^+) \le u_i$ and $\rho_i(S^- - \delta_i O_i) \ge l_i$,
where in Definition~\ref{def:stop_condition} we defined $\delta_i = \min(e_i, 1)$.
Recall that $u_i$ is distributed as $\rho_i(S) + \Exp{t}$ and $l_i$ 
as the height of the highest point in $\clopen{0, t} \x \clopen{0, \rho_i(S)}$.

We will now simplify these expressions to remove the $\delta_i$ term, 
instead sampling $e_i$ as $\Pois{\rho_i(S)t}$,
$u_i$ as $\rho_i(S) + \Exp t$, 
and $l_i$ as $\rho_i(S) - \Exp t$ independently.
Instead of sampling $e_i$ first and $l_i$ second, 
as we did in Procedures~\ref{procedure:tau:split} and~\ref{procedure:tau:resample},
we can sample $l_i$ first as $\max\pr{\rho_i - \Exp{t}, 0}$ and $e_i$ as 0 if $l_i = 0$ and $1 + \Pois{l_i t}$ otherwise.
Since $l_i \le \rho_i(S)$, the distribution of $e_i - \delta_i$ is stochastically dominated by $\Pois{\rho_i(S) t}$.

Since in the stop condition we check if $\rho_i(S^- - \delta_i O_i) < l_i$,
and since $\rho_i(S^- - \delta_i O_i) \ge 0$, 
$\rho_i(S^- - \delta_i O_i)$ cannot be smaller than $l_i$ when $l_i = 0$,
the condition is always satisfied when $l_i = 0$.
To simplify the distributions, we can thus sample $l_i$ as $\rho_i - \Exp{t}$ instead of 
$\max\pr{\rho_i - \Exp{t}, 0}$ without changing the probability of the reaction being stable.

We can thus simplify Condition~\ref{stop_condition_2} by 
replacing $S^- -\delta_i O_i^-$ with $S^-$
and $\delta_i$ with 0, and sampling $e_i$ as $\Pois{\rho_i(S)t}$.
This substitution does not affect~\ref{stop_condition_1}.

Thus, we have
\begin{align}
\rho_i(S^-) = r_i \prod_{k \in I_i} \binom{S^-[k]}{C_i[k]} \ge r_i \prod_{k \in I_i} \binom{S[k] - \Delta_k}{C_i[k]}, \label{lemma:density:lower}\\
\rho_i(S^+) = r_i \prod_{k \in I_i} \binom{S^+[k]}{C_i[k]} \le r_i \prod_{k \in I_i} \binom{S[k] + \Delta_k}{C_i[k]}. \label{lemma:density:upper}
\end{align}

We will now expand $\binom{S[k] + \Delta_k}{C_i[k]}$ to a polynomial in $\Delta_k$.
If we tried to expand $\pr{S[k] + \Delta_k}^{C_i[k]}$ instead of $\binom{S[k] + \Delta_k}{C_i[k]}$, 
we could have used the binomial theorem to expand it to 
$S[k]^{C_i[k]}\pr{\sum_i \binom{C_i[k]}{i}\pr{\frac{\Delta_k}{S[k]}}^i}$.
In the case of $\binom{S[k] + \Delta_k}{C_i[k]}$, we have a product of the form 
\begin{align*}
    \frac{\pr{S[k] + \Delta_k}\pr{S[k] - 1 + \Delta_k}\dots\pr{S[k] - C_i[k] + 1 + \Delta_k}}{C_i[k]!}.
\end{align*}
We expand the parentheses of the numerator, 
separating the $S[k] - j$ terms and the $\Delta_k$ terms,
yielding
\begin{align*}
    \frac{S[k](S[k] - 1)\dots(S[k] - C_i[k] + 1) + \Delta_k(S[k] - 1)\dots(S[k] - C_i[k] + 1) + \dots}{C_i[k]!}.
\end{align*}
Every term in this sum is a product of $(S[k] - j)$ terms and $\Delta_k$ terms.
We multiply and divide the expression by $\binom{S[k]}{C_i[k]}$ and cancel all $(S[k] - j)$ terms in the numerator, 
yielding $\binom{S[k]}{C_i[k]}$ times a sum of terms of the form $\frac{\Delta_k}{S[k] - a} \frac{\Delta_k}{S[k] - b}\dots$.
In each such term, since $S[k] \ge 2C_i[k]$, we have $S[k] - C_i[k] \ge \frac{S[k]}{2}$, 
so 
\begin{align*}
    \frac{\Delta_k}{S[k]} \le \frac{\Delta_k}{S[k] - j} \le 2\frac{\Delta_k}{S[k]}.
\end{align*}
We use this inequality to group terms together, leading to:
\begin{align*}
    \binom{S[k]}{C_i[k]}\pr{\sum_{i=0}^{C_i[k]} \binom{C_i[k]}{i}\pr{\frac{\Delta_k}{S[k]}}^i} \le 
    \binom{S[k] + \Delta_k }{C_i[k]} \le 
    \binom{S[k]}{C_i[k]}\pr{\sum_{i=0}^{C_i[k]} 2^i\binom{C_i[k]}{i}\pr{\frac{\Delta_k}{S[k]}}^i}.
\end{align*}

Since the $C_i[k]$ are in $O(1)$, we can simplify the expression to
\begin{align}
    \binom{S[k]}{C_i[k]}\pr{\sum_{i=0}^{C_i[k]} c_i \pr{\frac{\Delta_k}{S[k]}}^i} \le 
    \binom{S[k] + \Delta_k }{C_i[k]} \le 
    \binom{S[k]}{C_i[k]}\pr{\sum_{i=0}^{C_i[k]} c_i' \pr{\frac{\Delta_k}{S[k]}}^i},\label{eq:open_binomial}
\end{align}
with $c_0 = c_0' = 1$.

Rearranging the right-hand side of~\eqref{lemma:density:lower}
by opening the each binomial coefficient using~\eqref{eq:open_binomial} 
and grouping the terms by the powers of $\Delta_k$,
we deduce that there are constants $\cbr{c_1, c_2, \dots}$ 
such that the right-hand side of~\eqref{lemma:density:lower} is at least
\begin{align*}
    r_i \prod{\binom{S[k]} {C_i[k]}} \pr{
        1 
        - \sum_{k \in I_i} {c_k \frac{\Delta_k}{S[k]}} 
        + \sum_{k_1, k_2 \in I_i} c_{k_1}c_{k_2} \frac{\Delta_{k_1}\Delta_{k_2}}{S[k_1]S[k_2]} 
        + \dots}.
\end{align*}
Similarly, the right-hand side of~\eqref{lemma:density:upper} is at most
\begin{align*}
    r_i \prod{\binom{S[k]} {C_i[k]}} \pr{1 
    + \sum_{k \in I_i} {c_k \frac{\Delta_k}{S[k]}} 
    + \sum_{k_1, k_2 \in I_i}{c_{k_1}c_{k_2} \frac{\Delta_{k_1} \Delta_{k_2}}{S[k_1]S[k_2]}}
    + \dots}.
\end{align*}

Substituting $\rho_i(S) = r_i \prod{\binom{S[k]} {C_i[k]}}$, we get:
\begin{align*}
    \rho_i(S^-) \ge \rho_i(S) \pr{1 
    - \sum_{k \in I_i} {c_k \frac{\Delta_k}{S[k]}} 
    + \sum_{k_1, k_2 \in I_i}{c_{k_1}c_{k_2} \frac{\Delta_{k_1} \Delta_{k_2}}{S[k_1]S[k_2]}}
    + \dots},
\end{align*}
and
\begin{align*}
    \rho_i(S^+) \le \rho_i(S) \pr{1 
    + \sum_{k \in I_i} {c_k \frac{\Delta_k}{S[k]}} 
    + \sum_{k_1, k_2 \in I_i}{c_{k_1}c_{k_2} \frac{\Delta_{k_1} \Delta_{k_2}}{S[k_1]S[k_2]}}
    + \dots}.
\end{align*}

Since Lemma~\ref{lemma:full:sk_large} ensures that $\Delta_k \le S[k]$,
every term $\frac{\Delta_{k_1} \Delta_{k_2}\dots}{S[k_1]S[k_2]\dots}$ with more than one $\Delta_k$
is no more than the first term $\frac{\Delta_{k_1}}{S[k_1]}$ in the product.
Since there is only a finite number of them, we can substitute them with the first terms $\frac{\Delta_{k_1}}{S[k_1]}$
and increase the appropriate $c_k$, leading to the inequalities
\begin{align}
    \rho_i(S^-) \ge \rho_i(S) \pr{1 - \sum_{k \in I_i} {c_k \frac{\Delta_k}{S[k]}}}, 
    \label{eq:full:lower_bound}
\end{align}
and
\begin{align}
    \rho_i(S^+) \le \rho_i(S) \pr{1 + \sum_{k \in I_i} {c_k \frac{\Delta_k}{S[k]}}}.
    \label{eq:full:upper_bound}
\end{align}

The stability conditions $\rho_i(S^+) \le u_i$ and $\rho_i(S^-) \ge l_i$
are equivalent to $\rho_i(S^+) - \rho_i(S) \le u_i - \rho_i(S)$ and $\rho_i(S) - \rho_i(S^-) \le \rho_i(S) - l_i$.
By~\eqref{eq:full:lower_bound} and~\eqref{eq:full:upper_bound},
\begin{align}
    \rho_i(S^+) - \rho_i(S), \rho_i(S) - \rho_i(S^-) \le \sum_{k \in I_i} {c_k \frac{\Delta_k}{S[k]}}.
    \label{eq:density:combined_bound}
\end{align}

Since $U_i(S, t)$ is an indicator random variable, 
$P(U_i(S, t) = 1) = \EE(U_i(S, t))$.
Moreover, $U_i(S, t) = 1$ when $R_i$ is unstable.
Since $R_i$ is stable when~\eqref{eq:density:combined_bound} 
is smaller than two independent random variables distributed as $\Exp{t}$, we have:
\begin{align*}
    \EE\pr{U_i(S, t)} 
    & \le \EE\pr{1 - e^{-2\rho_i(S) tc \pr{\sum_{k \in I_i} {c_k \frac{\Delta_k}{S[k]}}}}} \\
    & \le 2\rho_i(S) tc \sum_{k \in I_i} {c_k \frac{\EE\pr{\Delta_k}}{S[k]}} \\
    & = 2\rho_i(S) t^2 c \sum_{k \in I_i} {c_k \frac{\EE\pr{\rho_k^\pm}}{S[k]}},
\end{align*}
proving the claim.

\subsubsection{A bound for \texorpdfstring{$U_i(S, t)$}{Ui(S, t)} when \texorpdfstring{$S[k]$}{S[k]} are small}

This section is devoted to proving a second bound that does not depend on $S[k]$, 
following the proof of Lemma~\ref{lemma:finite_unstable}.

\begin{lemma}\label{lemma:full:bad_complexity}
    Let $R_i$ be a reaction with input vector $I_i$, rate constant $r_i$, and rate $\rho_i$.
    Let $\cG_i$ be the set of reactions on which $R_i$ depends.
    Then there exists a constant $c$ that depends only on $\cR$, such that:
    \begin{align*}
        \EE\pr{U_i\pr{S, t}} \le \rho_i t^2 c \pr{\sum_{j \in \cG_i}\rho_j(S)}^2.
    \end{align*}
\end{lemma}

If one of the two following mutually exclusive conditions is satisfied, then $R_i$ is stable:
\begin{itemize}
    \item For all $j \in \cG_i$ we have $e_j = 0$.
    \item $\sum_{j \in \cG_i} e_j = 1$, and~\ref{stop_condition_1} is satisfied.
\end{itemize}

We have already shown in the proof of Lemma~\ref{lemma:finite_unstable}
that when $\sum_i e_i = 0$ the reaction $R_i$ is stable, 
and that when $\sum_i e_i = 1$ the reaction $R_i$ satisfies~\ref{stop_condition_2},
and thus, if it satisfies~\ref{stop_condition_1} as well, it is stable.

The probability of all reactions having zero occurrences is $e^{-t \sum_{j \in \cG_i} \rho_j(S)}$.
The probability of having exactly one occurrence in all reactions is $\pr{t \sum_{j \in \cG_i} \rho_j(S)}e^{-t \sum_{j \in \cG_i} \rho_j(S)}$.
When $\sum_i e_i = 1$, $\Delta_k$ are bounded by a constant, 
and hence by~\eqref{eq:full:upper_bound}, $\rho_i(S^+) = \bigo{\rho_i(S)}$.
Thus, the probability of satisfying~\ref{stop_condition_1} 
if $\sum_i e_i = 1$ is at least $e^{-\bigo{\rho_i(S)t}}$.

The total probability of the reaction being stable is thus:
\begin{align*}
    & e^{-t \sum_{j \in \cG_i} \rho_j(S)} + \pr{t \sum_{j \in \cG_i} \rho_j(S)}e^{-t \sum_{j \in \cG_i} \rho_j(S)}e^{-\bigo{\rho_i(S)t}} \\
    & = e^{-t \sum_{j \in \cG_i} \rho_j(S)}\pr{1 + \pr{t \sum_{j \in \cG_i} \rho_j(S)} e^{-\bigo{t\rho_i(S)}}} \\
    & \ge \pr{1 - t \sum_{j \in \cG_i} \rho_j(S)}\pr{1 + \pr{t \sum_{j \in \cG_i} \rho_j(S)} \pr{1 - \bigo{t\rho_i(S)}}} \\
    & = 1 - \bigo{t^2\pr{\sum_{j \in \cG_i}\rho_j(S)}\pr{\sum_{j \in \cG_i} \rho_j(S) + \rho_i(S)}} + \bigo{t^3}.
\end{align*}

Thus, $\EE\pr{U_i(S, t)}$ is smaller than $\bigo{t^2\pr{\sum_{j \in \cG_i}\rho_j(S)}^2}$.

\subsubsection{The final calculation}

\begin{lemma}
    Let $W_i(S, t)$ be the work the algorithm performs when processing $R_i$ 
    in a recursive step with initial state $S$ spanning a time segment of length $t$.
    Then, $W_i(S, t) = O(D_i)$.
\end{lemma}
\begin{proof}
    In every $\tau$-leaping step where $R_i$ is unstable, 
    all reactions depending on $R_i$ and on which $R_i$ depends must be active.
    For each of these reactions we perform a constant amount of work in every step of the 
    full $\tau$-splitting algorithm.
    Thus, $W_i(S, t) = O(D_i)$.
\end{proof}

\begin{lemma}
    The contribution of the reaction $R_i$ to the runtime density is in
    \begin{align*}
        \bigo{D_i \sum_{d = 0}^\infty \frac{1}{2^{-d}T}\overline U_i(S, 2^{-d}T)}.
    \end{align*}
\end{lemma}
\begin{proof}
    The contribution of $R_i$ to the runtime of the algorithm is:
    \begin{align}
        D_i \sum_{d = 0}^\infty \sum_{n = 0}^{2^d} \overline U_i\pr{S_{n2^{-d}T}, 2^{-d}T}, \label{lemma:density:real}
    \end{align}
    since we perform $\bigo{D_i}$ work for every recursive call where $R_i$ is unstable.
    We count the work required by recursive calls where $R_i$ is stable, 
    but is split because an unstable reaction $R_j$ depends on it, 
    as part of the contribution of $R_j$ to the runtime.

    The integral over the contribution of $R_i$ to the time density function we have defined is:
    \begin{align}
        D_i \sum_{d = 0}^\infty \frac{1}{2^{-d}T} \int_{t=0}^T \overline U_i\pr{S_t, 2^{-d}T} .\label{lemma:density:density}
    \end{align}
    The difference between~\eqref{lemma:density:real} and~\eqref{lemma:density:density} is that in~\eqref{lemma:density:density} 
    we have a continuous integral over time segments, 
    and in the real runtime~\eqref{lemma:density:real} we have a discrete sum over the dyadic splits of $\clopen{0, T}$.

    Let $t_0 = 2^{-d-1}T\ceil{\frac{t}{2^{-d-1}T}}$ be the rounding up of 
    $t$ to an integer multiple of $2^{-d-1}T$.
    Since $t_0 - t < 2^{-d-1}T$, we have $t_0 + 2^{-d-1}T < t + 2^{-d}T$, 
    and hence $\clopen{t_0, t_0 + 2^{-d-1}T} \subset \clopen{t, t + 2^{-d}T}$.
    When the $\tau$-leaping step spanning $\clopen{t, t + 2^{-d}T}$ is not split, 
    our sampling procedure guarantees that any $\tau$-leaping 
    step over a time segment contained in $\clopen{t, t + 2^{-d}T}$ is stable 
    (We have proved this for the dyadic splits in Procedures~\ref{procedure:tau:split} and~\ref{procedure:tau:resample}, 
    and the same proof works for any subsegment).
    Thus, $\EE\pr{\overline U(S_t, 2^{-d}T)} \ge \EE\pr{\overline U(S_{t_0}, 2^{-d-1}T) | S_t}$.
    Substituting this inequality into~\eqref{lemma:density:real}, we get:
    \begin{align}
        & \EE\pr{D_i\int_{t=0}^T \sum_{d = 0}^\infty 2^{-d}T \overline U_i\pr{S_t, 2^{-d}T}dt } \nonumber \\
        & = \EE\pr{D_i\sum_{d = 0}^\infty 2^{-d}T  \int_{t=0}^T \overline U_i\pr{S_t, 2^{-d}T}dt } \nonumber \\
        & \ge \EE\pr{D_i\sum_{d = 0}^\infty 2^{-d}T \int_{t=0}^T \overline U_i\pr{S_{2^{-d-1}T\ceil{\frac{t}{2^{-d-1}T}}}, 2^{-d-1}T}} dt \nonumber \\
        & = \EE\pr{D_i\sum_{d = 1}^\infty\sum_{n=0}^{2^d} \overline U_i\pr{S_{n2^{-d}T}, 2^{-d}T}} \nonumber \\
        & \ge \EE\pr{D_i\sum_{d = 0}^\infty\sum_{n=0}^{2^d} \overline U_i\pr{S_{n2^{-d}T}, 2^{-d}T}} - D_i \label{eq:density:calc}.
    \end{align}
    Since $D_i$ is constant and is negligible relative to the total work required by the simulation, 
    the term~\eqref{eq:density:calc} is equal up to an additive constant to~\eqref{lemma:density:real}, 
    the expected contribution of $R_i$ to the real runtime, 
    and the claim follows.
\end{proof}

We are now ready to complete the proof.

Recall that the total work density is:
\begin{align}
    D_i\sum_{d = 0}^\infty \frac{1}{2^{-d}T}\overline U_i(S, 2^{-d}T),\label{eq:full:work_density}.
\end{align}

We start by proving the bound~\eqref{eq:full:good_complexity}.
Under the theorem's hypothesis for~\eqref{eq:full:good_complexity}, we have:
\begin{align}
    U_i(S, 2^{-d}T) \le \min\pr{1, 2^{-2d}T^2\rho_i c \sum_{k \in I_i} \frac{\rho_k^\pm}{S[k]}}\label{eq:ui_inequality}
\end{align}

We separate the proof to two cases, depending on which term in the minimum is smaller.
Let 
\begin{align*}
    d^* = \ceil{\log_2\pr{T\sqrt{\rho_i c \sum_{k \in I_i} \frac{\rho_k^\pm}{S[k]}}}} = \log_2\pr{T\sqrt{\rho_i c \sum_{k \in I_i} \frac{\rho_k^\pm}{S[k]}}} + f,
\end{align*}
where $f \in [0, 1]$, be the depth where the second term in the minimum becomes lesser or equal to the first.

We split the sum~\eqref{eq:full:work_density} to the part where $d \le d^*$ and the part where $d > d^*$.
When $d \le d^*$, we have no better bound than $\overline U_i(S, 2^{-d}T) \le 1$,
yielding:
\begin{align*}
    D_i\sum_{d = 0}^{d^*} \frac{1}{2^{-d}T}\overline U_i(S, 2^{-d}T) 
    \le D_i\sum_{d = 0}^{d^*} \frac{1}{2^{-d}T}
    = 2 D_i \frac{1}{2^{-d^*} T} = 2 D_i \sqrt{\rho_i c \sum_{k \in I_i} \frac{\rho_k^\pm}{S[k]}}.
\end{align*}

We now turn to the second case 
\begin{align*}
    D_i\sum_{d = d^*}^{\infty} \frac{1}{2^{-d}T}\overline U_i(S, 2^{-d}T).
\end{align*}

To get a bound in terms of $U_i(S, t)$ instead of $\overline U_i(S, t)$,
we use Corollary~\ref{corollary:unstable_inactive_integral_inequality}, yielding
\begin{align*}
    \int_S D_i\sum_{d = d^*}^{\infty} \frac{1}{2^{-d}T}\overline U_i(S, 2^{-d}T) dS 
    & = \int_S D_i\sum_{d = d^*}^{\infty} \frac{1}{2^{-d}T} \sum_{d'=0}^\infty 2^{d'} U_i(S, 2^{-d-d'}T) dS \\
    & = \int_S D_i\sum_{d = d^*}^{\infty} \frac{d - d^* + 1}{2^{-d}T} U_i(S, 2^{-d}T) dS,
\end{align*}
where the last inequality follows by grouping together equal terms of $U_i(S, t)$.

Using~\eqref{eq:ui_inequality}, we get 
\begin{align*}
    D_i \sum_{d = d^*}^{\infty} \frac{d - d^* + 1}{2^{-d}T} U_i(S, 2^{-d}T)
    & \le D_i \sum_{d = d^*}^{\infty} \pr{d - d^* + 1}2^{-d}T\rho_i c \sum_{k \in I_i} \frac{\rho_k^\pm}{S[k]} \\
    &   = D_i \sum_{d = d^*}^{\infty} \pr{d - d^* + 1}2^{-(d - d^*)}2^{-d*}T\rho_i c \sum_{k \in I_i} \frac{\rho_k^\pm}{S[k]} \\
    &   = 4D_i 2^{-d^*}T\rho_i c \sum_{k \in I_i} \frac{\rho_k^\pm}{S[k]} \\
    &   = 4D_i \sqrt{\rho_i c \sum_{k \in I_i} \frac{\rho_k^\pm}{S[k]}}
\end{align*}

Therefore, the total runtime density is in $\bigo{D_i \sqrt{\rho_i \sum_{k \in I_i} \frac{\rho_k^\pm}{S[k]}}}$,
proving~\eqref{eq:full:good_complexity}.

We next prove the bound~\eqref{eq:full:bad_complexity}.
By Lemma~\ref{lemma:full:bad_complexity}, $\EE\pr{U_i(S, t)} \le c t^2\pr{\sum_{j \in D}\rho_j(S)}^2$.
Set 
\begin{align*}
    d^* = \ceil{\log_2\pr{T\pr{\sum_{j \in D}\rho_j(S)}^{-1}}}.
\end{align*}
Similar calculations to the ones above show that the total runtime density is in 
\begin{align*}
    \bigo{D_i \pr{\sum_{j \in D}\rho_j(S)}},
\end{align*}
proving the bound~\eqref{eq:full:bad_complexity}.

\section{Code Availability}

The code for the $\tau$-splitting algorithm is available at \url{https://github.com/ormorni/sde}.

\section*{Acknowledgements}
The authors would like to thank 
Professor Gadi Fibich for reading the drafts of the paper and for 
his useful comments that significantly improved the presentation.

\bibliography{main} 

\begin{thebibliography}{10}

\bibitem{gegenhuber2017fusing}
T.~Gegenhuber, L.~De~Keer, A.~S. Goldmann, P.~H. Van~Steenberge, J.~O. Mueller,
  M.-F. Reyniers, J.~P. Menzel, D.~R. D’hooge, and C.~Barner-Kowollik,
  ``Fusing light-induced step-growth processes with raft chemistry for
  segmented copolymer synthesis: a synergetic experimental and kinetic modeling
  study,'' {\em Macromolecules}, vol.~50, no.~17, pp.~6451--6467, 2017.

\bibitem{weinberger2005stochastic}
L.~S. Weinberger, J.~C. Burnett, J.~E. Toettcher, A.~P. Arkin, and D.~V.
  Schaffer, ``Stochastic gene expression in a lentiviral positive-feedback
  loop: Hiv-1 tat fluctuations drive phenotypic diversity,'' {\em Cell},
  vol.~122, no.~2, pp.~169--182, 2005.

\bibitem{dobrinevski2012extinction}
A.~Dobrinevski and E.~Frey, ``Extinction in neutrally stable stochastic
  lotka-volterra models,'' {\em Physical Review E—Statistical, Nonlinear, and
  Soft Matter Physics}, vol.~85, no.~5, p.~051903, 2012.

\bibitem{cuppen2013kinetic}
H.~Cuppen, L.~Karssemeijer, and T.~Lamberts, ``The kinetic monte carlo method
  as a way to solve the master equation for interstellar grain chemistry,''
  {\em Chemical reviews}, vol.~113, no.~12, pp.~8840--8871, 2013.

\bibitem{gillespie1977exact}
D.~T. Gillespie, ``Exact stochastic simulation of coupled chemical reactions,''
  {\em The journal of physical chemistry}, vol.~81, no.~25, pp.~2340--2361,
  1977.

\bibitem{anderson2007time_dependent}
D.~F. Anderson, ``A modified next reaction method for simulating chemical
  systems with time dependent propensities and delays,'' {\em The Journal of
  chemical physics}, vol.~127, no.~21, 2007.

\bibitem{cai2007delayed}
X.~Cai, ``Exact stochastic simulation of coupled chemical reactions with
  delays,'' {\em The Journal of chemical physics}, vol.~126, no.~12, 2007.

\bibitem{rijal2025differentiable}
K.~Rijal and P.~Mehta, ``A differentiable gillespie algorithm for simulating
  chemical kinetics, parameter estimation, and designing synthetic biological
  circuits,'' {\em ELife}, vol.~14, p.~RP103877, 2025.

\bibitem{thanh2016rssacr}
V.~H. Thanh, R.~Zunino, and C.~Priami, ``Efficient constant-time complexity
  algorithm for stochastic simulation of large reaction networks,'' {\em
  IEEE/ACM transactions on computational biology and bioinformatics}, vol.~14,
  no.~3, pp.~657--667, 2016.

\bibitem{ghosh2021blsssa}
D.~Ghosh and R.~K. De, ``Block search stochastic simulation algorithm (blsssa):
  a fast stochastic simulation algorithm for modeling large biochemical
  networks,'' {\em IEEE/ACM Transactions on Computational Biology and
  Bioinformatics}, vol.~19, no.~4, pp.~2111--2123, 2021.

\bibitem{burrage1996stochatic_runge_kutta}
K.~Burrage and P.~M. Burrage, ``High strong order explicit runge-kutta methods
  for stochastic ordinary differential equations,'' {\em Applied Numerical
  Mathematics}, vol.~22, no.~1-3, pp.~81--101, 1996.

\bibitem{cao2005slow}
Y.~Cao, D.~T. Gillespie, and L.~R. Petzold, ``The slow-scale stochastic
  simulation algorithm,'' {\em The Journal of chemical physics}, vol.~122,
  no.~1, p.~014116, 2005.

\bibitem{ahmadian2017hybrid}
M.~Ahmadian, S.~Wang, J.~Tyson, and Y.~Cao, ``Hybrid ode/ssa model of the
  budding yeast cell cycle control mechanism with mutant case study,'' in {\em
  Proceedings of the 8th ACM International Conference on Bioinformatics,
  Computational Biology, and Health Informatics}, pp.~464--473, 2017.

\bibitem{howard1993neuron_volume}
C.~V. Howard, G.~Jolleys, D.~Stacey, A.~Fowler, P.~Wall{\'e}n, and M.~A.
  Browne, ``Measurement of total neuronal volume, surface area, and dendritic
  length following intracellular physiological recording,'' {\em
  Neuroprotocols}, vol.~2, no.~2, pp.~113--120, 1993.

\bibitem{purves2018neuroscience}
D.~Purves, G.~J. Augustine, D.~Fitzpatrick, W.~C. Hall, A.-S. LaMantia, and
  L.~E. White, {\em Neuroscience}.
\newblock Oxford University Press, 6th~ed., 2018.

\bibitem{gillespie1992rigorous}
D.~T. Gillespie, ``A rigorous derivation of the chemical master equation,''
  {\em Physica A: Statistical Mechanics and its Applications}, vol.~188,
  no.~1-3, pp.~404--425, 1992.

\bibitem{gillespie2001approximate}
D.~T. Gillespie, ``Approximate accelerated stochastic simulation of chemically
  reacting systems,'' {\em The Journal of chemical physics}, vol.~115, no.~4,
  pp.~1716--1733, 2001.

\bibitem{thanh2020rssalib}
V.~H. Thanh, ``Rssalib: a library for stochastic simulation of complex
  biochemical reactions,'' {\em Bioinformatics}, vol.~36, no.~18,
  pp.~4825--4826, 2020.

\bibitem{liu2013fceri}
Y.~Liu, D.~Barua, P.~Liu, B.~S. Wilson, J.~M. Oliver, W.~S. Hlavacek, and A.~K.
  Singh, ``Single-cell measurements of ige-mediated fc$\varepsilon$ri signaling
  using an integrated microfluidic platform,'' {\em PloS one}, vol.~8, no.~3,
  p.~e60159, 2013.

\bibitem{barua2012computational}
D.~Barua, W.~S. Hlavacek, and T.~Lipniacki, ``A computational model for early
  events in b cell antigen receptor signaling: analysis of the roles of lyn and
  fyn,'' {\em The Journal of Immunology}, vol.~189, no.~2, pp.~646--658, 2012.

\bibitem{kingman1992poisson}
J.~F.~C. Kingman, {\em Poisson processes}, vol.~3.
\newblock Clarendon Press, 1992.

\end{thebibliography}
\bibliographystyle{ieeetr}

\end{document}